\documentclass[10pt, conference, letterpaper]{IEEEtran}
\IEEEoverridecommandlockouts
\usepackage{algorithm}
\usepackage{algpseudocode}
\usepackage{amsmath}
\usepackage{amsthm}
\usepackage{pifont}
\usepackage{xspace}
\usepackage{threeparttable}
\usepackage{color}
\usepackage{colortbl}
\usepackage{url}
\usepackage{subcaption}
\usepackage{paralist}
\usepackage{wrapfig}
\usepackage{multirow}
\usepackage{listings}
\usepackage{comment}
\usepackage{booktabs}
\usepackage{makecell}
\usepackage{boxedminipage}
\usepackage{amsmath}
\usepackage{graphicx}
\usepackage{subfloat} 
\usepackage{minibox}
\usepackage{pifont}
\usepackage{color}
\usepackage{xcolor}
\usepackage{subfiles}
\usepackage[utf8]{inputenc}
\usepackage[english]{babel}
\usepackage{xcolor}
\usepackage{ulem}
\usepackage{adjustbox}
\usepackage{cite}
\definecolor{ black}{RGB}{10, 10, 200}



\newcommand{\bheading}[1]{{\vspace{4pt}\noindent{\textbf{#1}}}}
\newcommand{\iheading}[1]{{\vspace{2pt}\noindent{\textit{#1}}}}
\newcolumntype{?}{!{\vrule width 1pt}}

\newcounter{note}[section]

\newcommand{\secref}[1]{\mbox{Sec.~\ref{#1}}\xspace}

\newcommand{\figref}[1]{\mbox{Fig.~\ref{#1}}}

\newcommand{\ignore}[1]{}

\newcommand{\ie}{\textit{i.e.}\xspace}
\newcommand{\eg}{\textit{e.g.}\xspace}

\newcommand{\sysname}{Orthrus\xspace}

\newcounter{packednmbr}

\newenvironment{packeditemize}{
\begin{list}{$\bullet$}{
\setlength{\labelwidth}{0pt}
\setlength{\itemsep}{2pt}
\setlength{\leftmargin}{\labelwidth}
\addtolength{\leftmargin}{\labelsep}
\setlength{\parindent}{0pt}
\setlength{\listparindent}{\parindent}
\setlength{\parsep}{1pt}
\setlength{\topsep}{1pt}}}{\end{list}}

\newtheorem{theorem}{Theorem}
\newtheorem{lemma}{Lemma}

\usepackage{amsmath}
\def\BibTeX{{\rm B\kern-.05em{\sc i\kern-.025em b}\kern-.08em
    T\kern-.1667em\lower.7ex\hbox{E}\kern-.125emX}}

\ifx\submission\undefined
    \newcommand{\Sadoghi}[1]{{\color{cyan}{\textbf{Sadoghi:} #1}}}
    \newcommand{\Chen}[1]{{\color{ black}[{\textbf{Chen:} #1}]}}
\else
    \newcommand{\Sadoghi}[1]{}
    \newcommand{\Chen}[1]{}
\fi

\begin{document}
\title{\sysname: Accelerating Multi-BFT Consensus through Concurrent Partial Ordering of Transactions (Extended Version)
}

\author{
	\IEEEauthorblockN{
		Hanzheng Lyu\IEEEauthorrefmark{1}, 
		Shaokang Xie\IEEEauthorrefmark{2}, 
		Jianyu Niu\IEEEauthorrefmark{3}, 
		Ivan Beschastnikh\IEEEauthorrefmark{4}, 
        Yinqian Zhang\IEEEauthorrefmark{3}, 
        Mohammad Sadoghi\IEEEauthorrefmark{2}, 
        Chen Feng\IEEEauthorrefmark{1} 
} 
\IEEEauthorblockA{University of British Columbia (\IEEEauthorrefmark{1}Okanagan Campus, \IEEEauthorrefmark{4}Vancouver Campus) \\\IEEEauthorrefmark{2}University of California, Davis, \IEEEauthorrefmark{3}Southern University of Science and Technology}

\IEEEauthorblockA{\IEEEauthorrefmark{1}\{hzlyu@student.ubc.ca, chen.feng@ubc.ca\}, \IEEEauthorrefmark{2}\{skxie,msadoghi\}@ucdavis.edu} \IEEEauthorrefmark{3}niujy@sustech.edu.cn, \IEEEauthorrefmark{4}bestchai@cs.ubc.ca
    }

\maketitle

\begin{abstract}
Multi-Byzantine Fault Tolerant (Multi-BFT) consensus allows multiple consensus instances to run in parallel, resolving the leader bottleneck problem inherent in classic BFT consensus. However, the global ordering of Multi-BFT consensus enforces a strict serialized sequence of transactions, imposing additional confirmation latency and also limiting concurrency. 
In this paper, we introduce \sysname, a Multi-BFT protocol that accelerates transaction confirmation through partial ordering while reserving global ordering for transactions requiring stricter sequencing. 
To this end, \sysname strategically partitions transactions to maximize concurrency and ensure consistency. Additionally, it incorporates an escrow mechanism to manage interactions between partially and globally ordered transactions. 
We evaluated \sysname through extensive experiments in realistic settings, deploying 128 replicas in WAN and LAN environments. 
Our findings demonstrate latency reductions of up to 87\% in WAN compared to existing Multi-BFT protocols.
\end{abstract}

\color{black}
\section{Introduction} \label{sec:intro} 
Byzantine Fault Tolerant (BFT) consensus has received renewed interest due to its adoption in blockchain applications~\cite{amiri2024bedrock,amiri2019caper,gupta2020resilientdb,gupta2019proof,ruan2021blockchains,dolev2023sodsbc}
ranging from cryptocurrency~\cite{algorand, hanke2018dfinity, QuestMarko, lens} to Decentralized Finance (DeFi)~\cite{werner2022sok}.
BFT consensus enables a network of replicas to agree on the same sequence of transactions, effectively mitigating double-spending attacks~\cite{karame2012double}\textemdash where a user might attempt to use the same asset in multiple transactions—even in the presence of Byzantine replicas that can behave maliciously.
Most BFT consensus protocols (\eg, PBFT~\cite{pbft1999}) employ a leader-based scheme, where a designated leader coordinates with other replicas (referred to as backups) to reach an agreement on its proposals (\eg, transactions). However, this leader-based approach faces a significant performance bottleneck~\cite{gai2021dissecting, avarikioti2023fnf, stathakopoulou2022state, gupta2021rcc}: the leader's workload scales linearly with the number of replicas, which in turn impacts system throughput and latency. 

To mitigate this bottleneck, Multi-BFT consensus protocols enable multiple leader-based consensus instances to run in parallel~\cite{MIR-BFT, avarikioti2023fnf, gupta2021rcc, stathakopoulou2022state, Ladon2025}. 
In these protocols, each replica simultaneously acts as the leader for one instance and a backup for others.  
Replicas run each instance to agree on a sequence of blocks, which are then ordered into a global sequence, as shown in~\figref{fig:motivation}a. 
By the global ordering, Multi-BFT consensus appears as a single BFT instance, but distributes workload among replicas to better utilize bandwidth and increase throughput.
For example, ISS~\cite{stathakopoulou2022state}, a state-of-the-art Multi-BFT protocol, demonstrates significant performance improvements over single-leader protocols, achieving up to 37$\times$ and 56$\times$ better throughput for PBFT and HotStuff, respectively, on a network of 128 replicas. RCC~\cite{gupta2021rcc}, another Multi-BFT protocol,  improves upon ISS by optimizing the recovery mechanism.

\begin{figure}[t]
\centering
\includegraphics[width=0.95\linewidth]{ 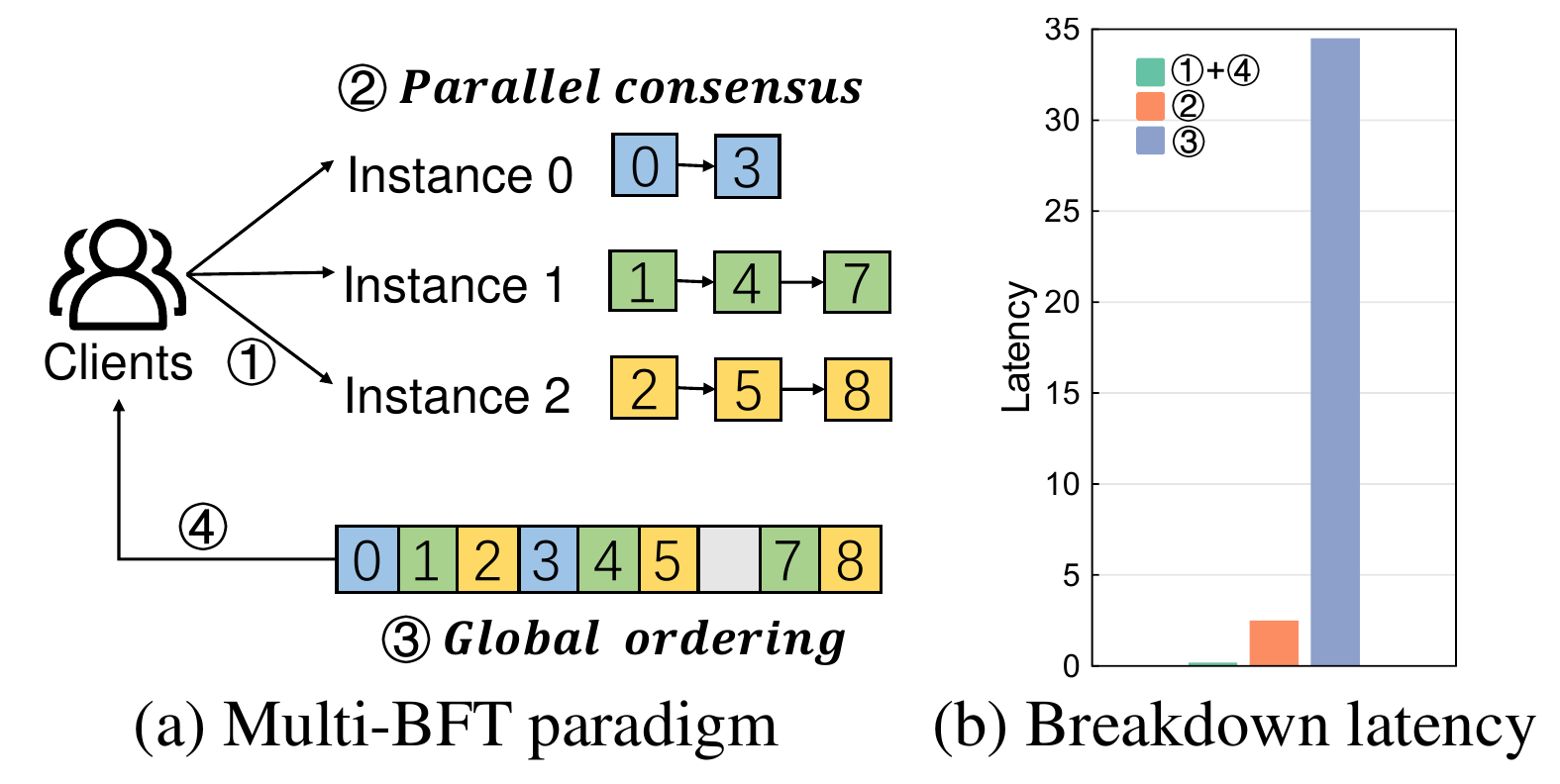}
\caption{\textbf{(a) Multi-BFT paradigm.} The $j$th block produced by instance $i$ is denoted as $B_j^i$. 
\textbf{(b) Breakdown latency with a straggler.} The green bar refers to transaction transmission delay (\ding{172} and \ding{175}), the orange refers to the delay of a block being delivered from consensus (\ding{173}), and the  black refers to the global ordering delay (\ding{174}).
}
\label{fig:motivation}
\end{figure}

\bheading{Our motivation.} The global ordering mechanism in Multi-BFT consensus, 
while essential for ensuring consistency, significantly limits concurrency and can increase latency. By enforcing a strict, serialized sequence across all transactions, global ordering prevents replicas from processing non-conflicting transactions concurrently, even when they target different records 
and could be processed in parallel. For example, consider two independent payment transactions: one in which Alice transfers funds from her account and another in which Bob transfers funds from his account.
In existing Multi-BFT consensus protocols~\cite{MIR-BFT, avarikioti2023fnf, gupta2021rcc, stathakopoulou2022state, Ladon2025}, both transactions would need to wait for a globally agreed-upon order before execution, blocking them from concurrent processing despite having no direct conflicts.

Second, global ordering increases the overall system latency, particularly when there are straggler instances that are much slower than other instances. As illustrated in~\figref{fig:motivation}a, if Instance 0 is a straggler, the missing block (\ie, block 6)  creates a gap in the global log, 
which then prevents subsequent blocks (\ie, blocks 7 and 8) from being executed until the missing block appears. 
As illustrated in~\figref{fig:motivation}b,  the latency breakdown for ISS~\cite{stathakopoulou2022state} with 16 replicas in a WAN setting
shows that global ordering delays can account for up to $92.8\%$ of the total latency when there is a straggler instance that is $10\times$ slower than other instances. Specifically, while the consensus instance delivers a transaction in about 2.5 seconds, the global ordering process takes significantly longer—an additional 34.5 seconds—highlighting the substantial delay caused by global ordering.
(Further evaluation results are presented in Sec.\ref{sec:evaluation}.) 

The expensive global ordering significantly affects the performance of Multi-BFT, especially with stragglers. Thus, this raises a vital question: 
\textit{How to minimize the use of global ordering in Multi-BFT consensus?}





\bheading{Our approach.} We observe that there are two primary classes of transactions in blockchain-based financial applications: \textit{conflict-free transactions }(\eg, payments {with different payers})~\cite{bentov2016cryptocurrencies, collins2020online} and \textit{general non-commutative transactions} (\eg, smart contracts {modifying the same record}). The conflict-free transactions inherently enable concurrent execution. Notably, conflict-free transactions account for a significant proportion of the overall transactions. For instance, more than 46\% of transactions in Ethereum are conflict-free payment transactions. To optimize performance, we present \sysname, which accelerates the processing of conflict-free transactions by introducing concurrency, while ensuring support for general non-commutative transactions, which may require stricter ordering due to dependencies between operations. Unlike traditional systems where a transaction can only be executed once it has received a global order and all preceding blocks in the global sequence have been executed, our approach allows for a more relaxed execution condition. 

\iheading{Conflict-free transactions.} We often observe that there are two types of operations in payment transactions: incremental (e.g., adding funds) and decremental (e.g., withdrawing funds). 
Incremental operations are inherently commutative, as they only increase the balance of the accounts. Decrements in different accounts are also commutative, as they do not affect each other's state. This means that many transactions can be executed concurrently without using strict ordering.

To accomplish this, \sysname assigns client transactions to instances based on the payer, ensuring that transactions from the same payer are processed within the same instance. This \textit{transaction partition} mechanism introduces concurrency and prevents double spending~\cite{karame2012double}. However, some transactions may involve multiple payers. In such cases, the transaction is assigned to multiple instances, each instance managing the actions of one payer, allowing the system to handle complex transactions without sacrificing concurrency. 
However, this approach also introduces a challenge in maintaining transaction atomicity, as some payers might succeed while others fail, potentially leaving the transaction in an incomplete state.

To ensure atomicity, we draw inspiration from the \textit{escrow transactions}~\cite{o1986escrow} to design an \textit{escrow mechanism} which temporarily reserves the transfer amounts for each payer in an escrow state.  The transaction is only confirmed if all escrow requests are successfully committed. If any escrow request fails, the reserved funds for all other escrow requests are also released, ensuring transactional atomicity and consistency across instances. 


\iheading{General non-commutative transactions.} Beyond conflict-free transactions, there are also general transactions—such as those involving smart contracts—that include non-commutative operations. 
For these, we enable confirmation through global ordering to ensure safety and consistency. 
This hybrid design allows \sysname to maximize concurrency for commutative transactions while maintaining the necessary cross-instance guarantees to support general non-commutative transactions. To facilitate seamless interaction between contract transactions (which require strict global ordering) and standard payment transactions, we implement the previously mentioned \textit{escrow mechanism}. This mechanism allows contract-related payments to be temporarily escrowed within the partial log, preventing them from delaying payment transaction confirmations. Consequently, payment transactions can proceed without interruption, while contract transactions maintain the necessary ordering constraints. 




We build an end-to-end prototype of \sysname in Go~\cite{golang} and conduct an extensive evaluation on AWS. We use a real-world dataset with 200,000 transactions from 18,000 active users on the Ethereum network. This dataset provides a more accurate representation of transaction behaviors, which leads to more valid and reliable results. 
We compare \sysname with the state-of-the-art Multi-BFT consensus, including ISS~\cite{stathakopoulou2022state}, Mir-BFT~\cite{MIR-BFT}, RCC~\cite{gupta2021rcc}, DQBFT~\cite{dqbft} and Ladon~\cite{Ladon2025} in terms of throughput and latency. 

\bheading{Our contributions.} To summarize, this paper makes
the following contributions: 

\begin{packeditemize}

\item We introduce \sysname, a Multi-BFT consensus that leverages partial ordering of transactions to enhance concurrency while maintaining consistency across instances. 

\item We propose a novel escrow mechanism to ensure transaction atomicity and seamless integration of partially ordered and globally ordered transactions.

\item We evaluate \sysname with extensive experiments over WAN and LAN with 8-128 replicas distributed across $4$ regions. In the WAN setting, \sysname achieves at most 87\% lower latency than other protocols on $128$ replicas with a straggler. The LAN setting exhibits a similar trend.
\end{packeditemize}

\section{Motivations}\label{sec:partial}
We first examine blockchain scenarios where partial ordering is sufficient for transactions and then highlight the necessity of hybrid ordering. 
Transactions in blockchains typically fall into two categories: \textit{payment} (\ie, conflict-free transactions) and \textit{contract transactions} (\ie, general transactions). Payment transactions are straightforward exchanges, by which clients transfer funds. Each payment transaction involves incremental and decremental operations on payers' and payees' accounts. Due to the commutative characteristic of these operations, partial ordering is sufficient. 

In contrast, contract transactions represent more complex interactions involving various operations beyond simple transfers. Specifically, the operations on state variables can be accessed by multiple clients.
Unlike payment transactions, contract transactions often include non-commutative operations where the order in which they are performed affects the result, such as assignments.
This distinction between payment and contract transactions highlights \sysname's dual approach: maximizing concurrency for payment transactions while employing global ordering for contract transactions to maintain system consistency. 

\subsection{Why Partial Ordering Works?}
In a typical payment transaction, decremental operations reduce the payers' account balance, while incremental operations increase the payees' account balance. By confining transactions that involve the same payer to a single instance, transactions across instances can be executed concurrently without global ordering. To better understand the benefits of partial ordering, let’s consider three simple single-payer, single-payee transactions: $tx_1$: Alice $\rightarrow$ Carol, $tx_2$: Bob $\rightarrow$ Carol, and $tx_3$: Alice $\rightarrow$ Bob, where $\rightarrow$ denotes the transfer tokens from the lefter to the righter. Based on the three transactions, we have the following two observations. 




\begin{packeditemize}
    \item \textit{Observation 1:} {If two transactions have different payers and neither transaction affects the balance of the other's payer, they can be executed concurrently.}
\end{packeditemize}

{
\bheading{Example (independence): Transactions $tx_1$ and $tx_2$ have different payers and do not affect each other's payer balance.} 
Even though both transactions share Carol as a payee, they have different payers, Alice and Bob. Since \(tx_1\) and \(tx_2\) consume tokens from different payer accounts, their execution order does not impact the outcome; both transactions succeed regardless of whether \(tx_1\) or \(tx_2\) executes first.
}

\begin{packeditemize}
    \item \textit{Observation 2:} {If one transaction may affect the balance of the other transaction's payer, they may need to be executed sequentially.}
\end{packeditemize}

{
\bheading{Example (dependency): The payee of $tx_3$ is the payer of $tx_2$.}  
The transaction $tx_3$ affects Bob's balance, and $tx_2$ depends on Bob's balance to succeed. If Bob has sufficient balance, $tx_2$ and $tx_3$ can be executed concurrently. If Bob does not have enough balance to support $tx_2$ before $tx_3$ executes, then $tx_3$ must be executed first to ensure that $tx_2$ can succeed. 
}

{
\bheading{Example (confliction): Transactions $tx_1$ and $tx_3$ have the same payer.} 
Since both transactions share Alice as the payer, they directly reduce her account balance. If Alice has sufficient balance, $tx_1$ and $tx_3$ can be executed concurrently. If Alice has limited funds, the order of \(tx_1\) and \(tx_3\) becomes crucial: executing \(tx_1\) first may deplete her balance, resulting in a successful transfer to Bob but causing $tx_3$ to fail. Conversely, if $tx_3$ is executed first, $tx_1$ may fail due to insufficient funds.
}

{
These examples illustrate that we can enforce effective partial ordering to manage transaction dependencies and prevent conflicts. Transactions involving the same payer are assigned to the same instance and executed sequentially within that instance, ensuring consistency. For dependencies across instances, \sysname establishes a cross-instance partial order by generating references between dependent transactions. These designs guarantee that if one transaction affects the payer's balance in a way that impacts another transaction's execution, the two transactions are processed in the correct sequence to preserve consistency.
}

For transactions involving multiple payers and payees, each transaction can be split into multiple single-payer, single-payee sub-transactions, thereby adhering to the above partial ordering rules. To guarantee atomicity in multi-payer transactions, \sysname employs an escrow mechanism discussed in a later section.
This incurs minimal overhead, as each sub-transaction can be processed concurrently, and the escrow mechanism adds only a slight processing cost. 

\subsection{Why Needs {Hybrid} Ordering?}
While partial ordering suffices for payment transactions, contract transactions often involve complex, non-commutative operations (\eg, assignments) that require stricter ordering. The outcome of these operations depends on execution order, necessitating global ordering to maintain consistent results.



\begin{packeditemize}
    \item \textit{Observation 3:} Transactions {modifying the same record} 
    must be executed sequentially to preserve consistency.
\end{packeditemize}

\bheading{Example: Transactions {$tx_1$ and $tx_2$ modify the same record}.} {Assume that transaction \(tx_1\) assigns \(value_1\) to record \(o\), while transaction \(tx_2\) assigns\(value_2\) to record \(o\).} Here, the order between \(tx_1\) and \(tx_2\) is crucial: executing \(tx_1\) first and then \(tx_2\) results in \(o\) having \(value_2\), while reversing the order leaves \(o\) with \(value_1\).

Although transactions {modifying distinct records} can execute concurrently without ordering, most contract transactions involve non-commutative operations that inherently depend on execution order. 
Therefore, a global ordering mechanism is essential to maintain consistency across replicas and ensure that all contract transactions yield the same result. 

{
The above examples underscore the necessity of hybrid ordering in \sysname. For payment transactions, partial ordering maximizes concurrency, while contract transactions with complex operations rely on global ordering to ensure consistency across replicas. 
}




\section{System Model and Properties
} \label{sec:model}
\subsection{System Model}\label{subsec:systemmodel}
We consider a system composed of $n = \{r_i\}_{i=0}^{n-1}$ replicas, denoted as the set $\mathcal{N}$.  We assume a subset of up to $f$ replicas are \textit{Byzantine}, represented by the set $\mathcal{F}$, where $|\mathcal{F}| \leq f$ and $n \ge 3f+1$. The remaining replicas, denoted as $\mathcal{H} = \mathcal{N} \setminus \mathcal{F}$, are honest and strictly follow the protocol. All Byzantine replicas are assumed to be controlled by a single adversary, which is computationally bounded and cannot break cryptographic primitives.
Each replica $r_i$ has a public/private key pair $(pk_i, sk_i)$ established through a public-key infrastructure (PKI), enabling them to sign and verify messages.

\bheading{Network model.} Honest replicas are connected through authenticated point-to-point channels.
We adopt the partial synchrony model introduced by Dwork et al. \cite{dwork1988consensus}, commonly used in BFT consensus protocols \cite{pbft1999, hotstuff}. This model defines an unknown Global Stabilization Time (\textsf{GST}), after which the system behaves synchronously with a known message delivery bound $\Delta$. Formally, for any two honest replicas $r_i, r_j \in \mathcal{H}$, and any time $t \geq \textsf{GST}$, a message sent from $r_i$ to $r_j$ at time $t$ will be delivered by time $t + \Delta$. 

\subsection{Data Model}\label{datamodel}

\bheading{Object.} We follow the object-centric design~\cite{weihl1988commutativity,blackshear2023sui}, in which objects provide operations that can be called by transactions to examine and modify the object’s state. These objects are long-lived, like accounts, and are represented as $o = (key, value, op, con, type)$. Here, $key$ is a cryptographically unique identifier, and $value$ denotes the object’s current state, which can be updated as operations $op$ are performed. The $op$ attribute specifies the operation to be executed on the object (\eg, increment or decrement). The $con$ attribute sets a condition that must be satisfied following the execution of any operation.
The $type$ attribute categorizes objects as either \textit{owned} or \textit{shared}. 

\iheading{1) Owned objects.} They are associated with a specific owner, with $key$ corresponding to the owner’s address. Owned objects support two operations: incremental and decremental. The decremental operation requires the owner’s authorization through a digital signature. For example, Alice’s account object, which holds her balance, is an owned object. If Alice wishes to transfer tokens, she must authorize the transaction with her digital signature.

\iheading{2) Shared objects.} They have no specific owner and can be processed by anyone with authorization in the smart contract. Shared objects may support additional operations such as complex state changes or contract-specific actions.

\bheading{Transaction.} Transactions are defined as $tx = (O, id, \sigma)$, where $O$ denotes the set of objects $o$ involved in $tx$, each specifies an involved object and the specified operation $op$ to be performed. The $id$ field provides a unique identifier for each transaction. $\sigma$ is the signature, included for owned objects requiring authorization. Each transaction involves at least one owned object, as every transaction must be initiated by a client, whose account is classified as an owned object. 

Transactions are categorized into two types: payment and contract transactions. Payment transactions involve only owned objects, while contract transactions may involve both objects, allowing for greater flexibility in executing complex operations under a smart contract framework.

\bheading{Block.} A block is a tuple $b = (txs, ins, sn, S, \sigma)$, where: \( txs \) represents a batch of transactions and $ins$ indicates the specific instance to which the block belongs, which helps in identifying which instance is responsible for processing the transactions in the block. $sn$ denotes the sequence number of the block, which maintains the order of blocks within an instance. $S$ denotes the system state (defined in \secref{subsec:goal}) on which the block depends, which suggests the state under which the transactions within the block can be executed successfully. $\sigma$ is the cryptographic signature on $B$, which provides authenticity and integrity to the block.

\subsection{Core Components}\label{subsec:preliminaries}
\bheading{Sequenced broadcast (SB).} Sequenced broadcast is a critical variant of Byzantine total order broadcast \cite{cachin2011introduction} that ensures the total ordering of transactions and preserves the consistency among honest replicas.  The SB operates through two fundamental primitives: \textit{broadcast} and \textit{deliver}. A replica acts as the leader to broadcast a transaction with a sequence number, and all replicas collaborate to deliver the transaction with the sequence number. There is a failure detector that can detect leader failures.
SB guarantees that an honest replica will eventually deliver a transaction for every sequence number $sn$ (\ie, \textit{termination property}) and that all honest replicas will deliver the same transaction with the same $sn$ (\ie, \textit{agreement property}). 
In this work, we utilize SB protocol as a black box that inputs client transactions and outputs a sequence of delivered transactions. 

\bheading{Escrow method.} The escrow method is a concurrency control and data consistency mechanism originally developed to handle high-contention scenarios in database systems, particularly when aggregate quantities, such as inventory levels, must be managed under high transactional loads with minimal locking \cite{o1986escrow}. 
In the Escrow method, operations are divided into incremental and decremental actions, each modifying a resource quantity without directly conflicting with others. By assigning each transaction an escrow balance, the system can accommodate multiple requests on the same resource simultaneously, as long as the escrowed quantities suffice for each transaction's demand. If a transaction aborts, the escrowed quantity is simply returned to the available pool. 
In this work, we leverage the Escrow method to handle concurrent transaction execution in replicated state environments.

\subsection{System Goals}\label{subsec:goal}
We consider a Multi-BFT system composed of \(m\) SB instances, indexed from $0$ to $m-1$, which take clients' transactions as input. Each instance has a leader who broadcasts transactions with sequence numbers and coordinates all replicas to deliver a sequence of transactions. Transactions are executed after being partially ordered within an SB instance or globally ordered across instances. A transaction is confirmed once it is executed, either successfully or unsuccessfully. The system’s state \( S \) can be represented as a tuple, where each element corresponds to the maximum sequence number \( sn \) of an SB instance. 
Specifically, the state of the Multi-BFT system is given as \( S = (sn_0, sn_1, \ldots, sn_{m-1}) \), where \( sn_i \) denotes the maximum sequence number for the instance indexed by \( i \). The state of the system in the view of replica $r_i$ is denoted as $S_{r_i}$. 
The system must ensure the following properties:

\begin{packeditemize}
\item \textbf{Safety}: 
If two honest replicas reach the same state \( S \), they must have consistent values for all objects. 



\item \textbf{Liveness}: Given a transaction \( tx \) from
a correct client \( c \), $tx$ is eventually confirmed by all honest replicas.
\end{packeditemize}


\begin{figure}[t]
	\centering
    \includegraphics[width=1\linewidth]{ 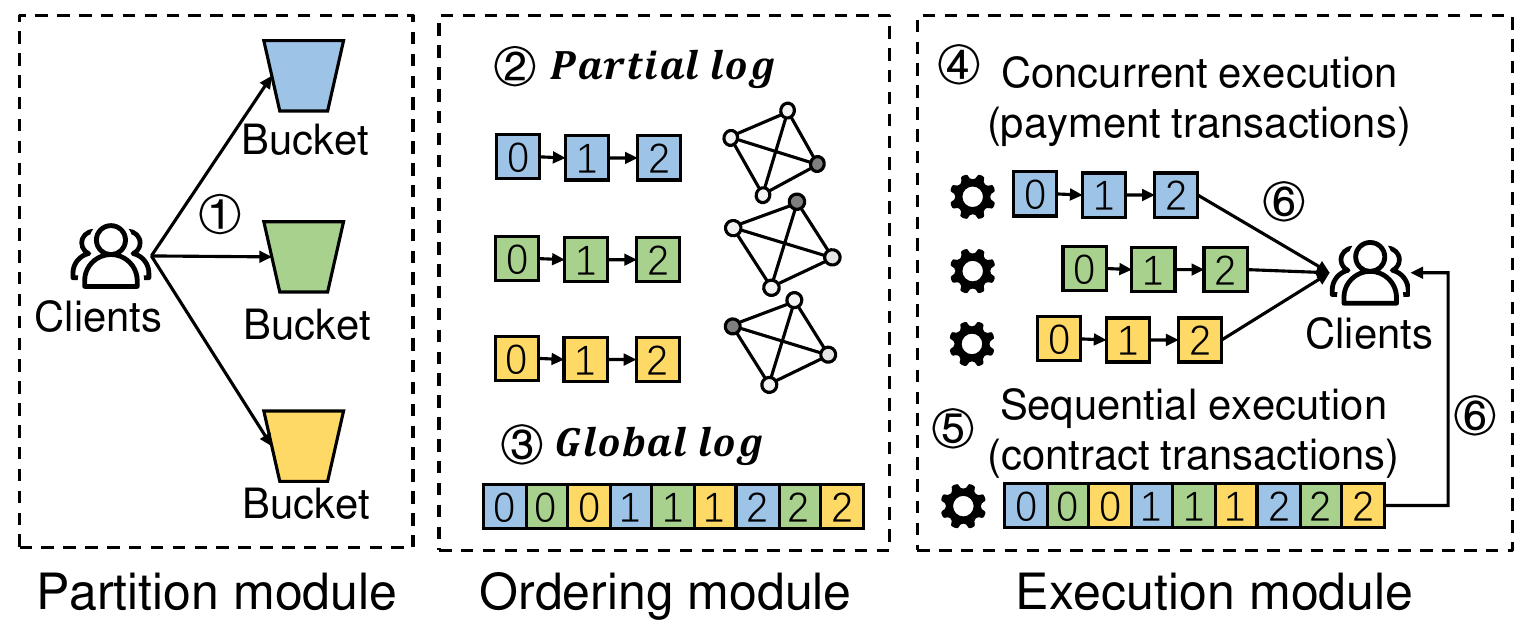}
	\caption{\textbf{An overview of \sysname.} Transactions are partitioned into distinct buckets, ordered according to their types, and finally executed.} 
	\label{fig:overview}
\end{figure}

\section{\sysname Overview} \label{sec:overview}

\subsection{System Architecture} \label{subsec:architecture}
The architectural overview of \sysname is provided in \figref{fig:overview}, which consists of three modules: partition module, ordering module, and executions module.

\bheading{Partition module.} It employs the bucket mechanism first introduced by Mir-BFT~\cite{MIR-BFT} to partition client transactions into buckets, with each bucket assigned to a unique instance. Each transaction is assigned to one or more buckets based on the objects it contains. 
Specifically, we define an assignment function $\mathsf{assign}$ that maps each owned object to a specific bucket. 
For example, $\mathsf{assign}$ may compute a bucket index for an object $o$ as the hash of the $key$ modulo $m$, where \( m \) is the total number of instances. Furthermore, the function can also be designed to balance loads across instances and minimize cross-instance interactions.
Transactions containing decremental operations on owned objects are assigned to the corresponding buckets of those objects. 

\bheading{Ordering module.} It contains multiple SB instances and a global ordering algorithm. Each instance independently selects a set of transactions from its corresponding bucket as input and outputs a sequence, or \textit{partial log} of transactions, agreed upon by all replicas within that instance. Transactions in a partial log are \textit{partially ordered}, meaning they are ordered within an instance. Payment transactions can be output in partial logs as soon as they are partially ordered, while contract transactions require a global order. The global ordering algorithm merged the partial logs from all instances into a \textit{global log}, which imposes a \textit{global order} across all transactions in the system. 

\bheading{Execution module.} It processes the ordered transactions produced by the ordering module. Payment transactions from each partial log can be executed concurrently, leveraging their partial ordering for parallel processing. In contrast, contract transactions in the global log are executed sequentially to maintain consistency across all instances. The execution results are then sent back to the clients.

\subsection{Transaction Workflow}
Each transaction in \sysname follows a workflow through the partition, ordering, and execution modules. Below, we outline the processing flow for a single transaction in \sysname.

\bheading{\ding{172}} Upon receiving a transaction, a replica assigns it to one or more buckets in the partition module, based on the involved owned objects (\secref{sec:partitionAlgorithm}).

\bheading{\ding{173}} The transaction is then packed into blocks and broadcast within the SB instances, which eventually deliver it with a designated sequence number (\secref{sec:orderAlgorithm}).

\bheading{\ding{174}} Contract transactions delivered by the instances are merged into a global log to establish a global order, while payment transactions can bypass this step (\secref{sec:orderAlgorithm}).

\bheading{\ding{175}} Payment transactions from partial logs are executed concurrently (\secref{sec:executeAlgorithm}).

\bheading{\ding{176}} Contract transactions in the global log are executed sequentially (\secref{sec:executeAlgorithm}).

\bheading{\ding{177}} Once successfully executed, the transaction is confirmed, and a response is sent to the client.


\subsection{Challenges and Solutions}\label{subsec:challenge}

This section discusses two key challenges of the hybrid ordering design of \sysname and corresponding solutions.

\bheading{Challenge-I: Ensuring atomicity of transactions.} A payment transaction involving multiple owned objects with decremental operations (\ie, multiple payers) may be assigned to multiple instances according to the partition module. If one payer executes the payment in their instance successfully while another payer fails in theirs, the transaction loses atomicity, leading to inconsistent states across instances. Maintaining atomicity across multiple instances, therefore, poses a significant challenge.

\bheading{Solution-I: Escrow mechanism for atomicity.} We design an escrow mechanism inspired by the escrow method proposed in~\cite{o1986escrow}, but simplify the escrow process by focusing on the temporary reservation of funds.
When a transaction with multiple payers is assigned to different instances, each instance performs an escrow operation on its respective payer’s balance, temporarily deducting the amount from each payer. If all escrow requests within the transaction are successfully committed, these temporary deductions become permanent. However, if any escrow request fails, all escrows associated with this transaction are canceled, and the amounts are refunded to the respective payers. Once a replica observes that all payers in a transaction have successfully escrowed their amounts, it commits all escrow requests. If any escrow request fails, it aborts all escrowed requests for the transaction, ensuring atomicity is preserved.

\bheading{Challenge-II: Avoiding payment transactions blocking.} When a contract transaction and a subsequent payment transaction share the same payer, the contract transaction may block the payment transaction due to the need for global ordering. Since the contract transaction must be globally ordered for consistency across all instances, any dependent transactions involving the same payer are delayed until the contract transaction is confirmed. This dependency leads to inefficiencies and increased latency for the payment transaction.

\bheading{Solution-II: Escrow mechanism to avoid blocking.}  
To address this issue, we extend the escrow mechanism to allow payment transactions to proceed without being blocked by pending contract transactions. When a replica encounters a contract transaction involving a payer, it escrows the required amounts from the payer. This approach allows subsequent payment transactions involving the same payer to be evaluated and processed as though the contract transaction’s decremental operation has been executed, thus avoiding blocking.
Once the contract transaction reaches global order and is confirmed, the replica either commits the escrowed amounts if the transaction executes successfully or refunds the amounts if it fails. This solution ensures that contract transactions do not block subsequent payments, improving system efficiency while maintaining consistency.

\begin{algorithm}[!t]
\caption{\sysname Algorithm for replica $r$}
\label{algorithm:main}
\begin{algorithmic}[1]
\State \textbf{upon} initialize system
\State\quad \textbf{for} $i \in [0, m-1]$ \Comment{\textcolor{purple}{$m$ is the number of instances}}
\State\quad\quad \textbf{if} $\mathsf{isLeader}$($i, r$) \textbf{then} \Comment{\textcolor{purple}{$r$ is the leader of $instance_i$}} 
\State\quad\quad\quad \textbf{for} {$sn \in \{0,1,2,...\}$}
\State\quad\quad\quad\quad $b.S \gets \mathsf{currentState}$
\State\quad\quad\quad\quad $b.txs \gets \mathsf{pullValidTx}(bucket_i, b.S)$
\State\quad\quad\quad\quad $b.ins \gets i$; $b.sn \gets sn$; $b.\sigma \gets \mathsf{sign}(b,r)$
\State\quad\quad\quad\quad \textbf{trigger}$\langle\mathsf{sb\text{-}broadcast}|b\rangle$  
\State

\State \textbf{upon} receive $tx$ \Comment{\textcolor{purple}{partition transactions}}
\State\quad \textbf{if} $\mathsf{validateTx}(tx)$ \textbf{then} 
\State\quad\quad \textbf{for} $o \in tx.O$  \textbf{where}   \Statex\quad\quad $o.type = owned \wedge o.op = decremental$  
\State\quad\quad\quad $i \gets \mathsf{assign}(o)$ 
\State\quad\quad\quad $\mathsf{push} (tx, bucket_i)$ 
\State

\State \textbf{upon event} $\langle\mathsf{sb\text{-}deliver}|b\rangle$ \Comment{\textcolor{purple}{order transactions}}
\State\quad $plog[b.ins][b.sn] \gets b$
\State\quad $\mathsf{globalOrder}$($b, glog$)

\State

\State \textbf{upon} $\mathsf{firstPending}(plog[i]) \neq \bot$ \Comment{\textcolor{purple}{execute transactions}}
\State \quad $tx \gets \mathsf{firstPending}(plog[i])$
\State \quad \textbf{for} $o \in tx.O$   \textbf{where}   
\Statex \quad $\mathsf{assign}(o) = i\wedge o.type = owned \wedge o.op = decremental$ 
\State \quad\quad $\mathsf{escrow}(o,tx)$ 
\State \quad\quad \textbf{if} $(o,tx) \notin elog$
\State \quad\quad\quad  $\mathsf{abortEscrow}(tx)$
\State \quad\quad\quad $\mathsf{abort}(tx)$ \Comment{\textcolor{purple}{remove $tx$ from all logs}}
\State \quad\quad \textbf{else if} $\mathsf{allEscrowed}(tx) \gets true$ 
\Statex \quad\quad\quad\quad\quad $\wedge tx.type = payment$
\State \quad\quad\quad $\mathsf{commitEscrow}(tx)$
\State \quad\quad\quad  \textbf{for} $o \in tx.O \ \textbf{where} \ o.op = incremental$ 
\State \quad\quad\quad\quad  $o.value \gets \mathsf{apply}(o,tx.o.op)$ 
\State

\State \textbf{upon} $\mathsf{firstPending}(glog) \neq \bot$ 
\State \quad $tx \gets \mathsf{firstPending}(glog)$
\State \quad \textbf{if} $\mathsf{isLastPosition}(tx, glog)$
\State \quad\quad \textbf{if} $\mathsf{exe}(tx) = true \wedge \mathsf{allEscrowed}(tx) = true$
\State \quad\quad\quad $\mathsf{commitEscrow}(tx)$
\State \quad\quad \textbf{else}
\State \quad\quad\quad $\mathsf{abortEscrow}(tx)$
\State \quad\quad\quad $\mathsf{abort}(tx)$ \Comment{\textcolor{purple}{remove $tx$ from $glog$}}
\State \quad \textbf{else} 
\State \quad\quad $\mathsf{remove}(tx,glog)$ \Comment{\textcolor{purple}{remove this $tx$ from $glog$}}

\end{algorithmic}
\end{algorithm}
\section{\sysname Algorithm} \label{sec:algorithm}
In this section, we present the core algorithm of \sysname, as illustrated in Algorithm~\ref{algorithm:main}. Overall, \sysname operates in epochs, following the design in~\cite{stathakopoulou2022state, Ladon2025}, with specific sequence numbers assigned to each instance in each epoch. A replica ends an epoch only after confirming all sequence numbers assigned to it within that epoch.
In each epoch, \sysname handles several key processes: transaction partitioning (\secref{sec:partitionAlgorithm}), partial and global ordering (\secref{sec:orderAlgorithm}), and transaction execution (\secref{sec:executeAlgorithm}). At the end of each epoch, it creates checkpoints and performs garbage collection (\secref{sec:checkpoint}).

\subsection{Partition Transactions}\label{sec:partitionAlgorithm}
Upon receiving a transaction \( tx \), replica \( r \)  adds \( tx \) to one or more buckets based on its involved objects.  First, it verifies the validity of the transaction’s format and checks the owner's signature. 
Once verified, replica \( r \) calls the \( \mathsf{assign}(o) \) function to determine the appropriate bucket(s) for \( tx \). Specifically, the replica iterates through each object involved in the transaction. For each object, if the object is of type \textit{owned} and the transaction's operation on the object is decremental, the replica calls \( \mathsf{assign}(o) \)  to determine the appropriate bucket index for the object. After obtaining the bucket index, \( tx \) is pushed into the corresponding bucket. Each bucket is an append-only list for backups but allows both push and pull operations for its leader. If \( tx \) is already in a bucket, it will not be added again to avoid duplication.

\subsection{Order Transactions}\label{sec:orderAlgorithm}
We employ two types of logs for transaction ordering: the partial log, denoted as \( plog \), and the global log, denoted as \( glog \). Each log consists of multiple entries, with each entry capable of holding a batch of transactions. Each instance maintains its \( plog \), which is used for the partial ordering of transactions within that instance. Meanwhile, the system as a whole maintains a single \( glog \) to ensure a global ordering of transactions across the system. 

We treat the SB as a black box, interacting with it through two defined events: $\langle\mathsf{sb\text{-}broadcast}|b\rangle$ and $\langle\mathsf{sb\text{-}deliver}|b\rangle$. Here, $\langle\mathsf{sb\text{-}broadcast}|b\rangle$ represents the event where a block $b$ is broadcasted into the SB, while $\langle\mathsf{sb\text{-}deliver}|b\rangle$ signifies the event when the SB delivers the block $b$ after ordering. 

\begin{algorithm}[t]
\caption{Escrow method}
\label{algorithm:escrow}
\begin{algorithmic}[1]
\State \textbf{function} $\mathsf{escrow}(o, tx)$
\State \quad $value \gets \mathsf{apply}(o, o.op)$
\State \quad \textbf{if} $value \geq o.con$
\State \quad\quad $o.value \gets value$
\State \quad\quad $\mathsf{append}((o, tx), elog)$
\State

\State  \textbf{function} $\mathsf{allEscrowed}(tx)$
\State \quad $\mathsf{allEscrowed}(tx) \gets true$
\State \quad \textbf{if} $\exists \ o \in tx.O \ \textbf{s.t.} \ ((o.type = \textit{owned})  \wedge$  
\Statex \quad\quad $ o.op = \textit{decremental})\wedge ((o, tx) \notin elog)$
\State \quad\quad $\mathsf{allEscrowed}(tx) \gets \mathsf{false}$

\State
\State \textbf{function} $\mathsf{commitEscrow}(tx)$
\State \quad \textbf{for} $o \in tx.O \wedge (o,tx) \in elog$ 
\State \quad\quad\quad $\mathsf{remove}((o, tx), elog)$
\State

\State \textbf{function} $\mathsf{abortEscrow}(tx)$
\State \quad \textbf{for} $o \in tx.O \wedge (o,tx) \in elog$ 
\State \quad\quad $\mathsf{undo}(o, o.op)$
\State \quad\quad $\mathsf{remove}((o, tx), elog)$

\end{algorithmic}
\end{algorithm}
\bheading{Broadcast transactions.} Upon initializing \sysname, for each instance, if replica $r$ is the leader of $instance_i$, it enters a loop 
in which it creates a block $b$ for each sequence number $sn$ in the current epoch and then broadcasts $b$ in the SB instance. 
In each iteration, the leader creates a block \(b\) and references the current system state $S$, which consists of the sequence number of the last block on each instance. Using this state $S$, the leader pulls a specified number of the oldest transactions from the instance's bucket. These transactions are valid under \(S\), meaning that \(S\) provides a consistent baseline against which each transaction's prerequisites are satisfied. Backup replicas can check the validation of transactions based on this state.
If there are insufficient transactions to meet this number, the leader waits for a timeout and then pulls available transactions from the bucket. 
The leader then sets the instance index and sequence number for the block, signs it, and broadcasts it in the SB instance, where all replicas participate in ordering and delivering the block. If \(r\) is not the leader, it participates in the SB instances as a backup. 

\bheading{Deliver transactions.} 
Upon delivering a block $b$ from an SB instance with a sequence number, the replica orders the block by appending it to the partial log \( plog \) at the instance index \( b.ins \) and sequence number \( b.sn \) as well as invokes the \( \mathsf{globalOrder} \) function to append it to the global log $glog$.  
{Each replica locally computes the global index of a block based on its parameters, which determine its position in the global order. To achieve this, \sysname adopts the dynamic global ordering algorithm from Ladon~\cite{Ladon2025}, which is detailed in Appendix~\ref{appen:globalorder}.}

\bheading{Failure detector.} In \sysname, a failure detection module (also called the view-change mechanism) is integrated into the SB protocol. This module enables replicas to change the faulty leader for an SB instance and has been widely used in prior work~\cite{castro2002practical, hotstuff, FireLedger, gueta2019sbft}. 
For instance, in PBFT, replicas begin a \textit{view change} to replace the leader when suspecting the leader of Byzantine behavior. When $l_i$, the leader of $instance_i$ fails at sequence number $sn$, the recovery process involves three steps:
1) All honest replicas detect the failure of $l_i$.
2) All honest replicas agree on the state of $instance_i$ and the new leader $l^{'}_i$.
3) All honest replicas restart $instance_i$ from $sn$.

The failure detector is used to deal with censorship attacks and spoofing attacks, In a censorship attack, Byzantine leaders can selectively ignore or censor certain transactions. 
To address this, we require that a client broadcast the transaction $tx$ to at least $f+1$ replicas, where $f$ is the maximum number of faulty replicas. This ensures that at least one honest replica will receive and push $tx$ to the corresponding bucket. If the leader fails to propose $tx$ within a reasonable period, other replicas can detect this failure and request a leader replacement. 
In a spoofing attack, a Byzantine leader might broadcast a block referencing an incorrect state $S$. If a replica finds that a transaction within the block is invalid based on the specified state $S$, it detects the leader's failure. The replica can also request the blocks it missed in the state $S$ from the leader. Should the leader be unable to provide these blocks, it detects the leader's failure as well. Additionally, we limit the number of messages that one replica can send to another within a given period to prevent Denial of Service (DoS) attacks.

\subsection{Execute Transaction}\label{sec:executeAlgorithm}
To introduce the transaction execution process, we first provide an overview of the escrow mechanism, a foundational component of \sysname’s execution phase.

\bheading{Escrow mechanism.}  Each replica maintains an escrow log $elog$ to manage the escrow of objects and transactions. The Escrow mechanism in \sysname is shown in Algorithm~\ref{algorithm:escrow}.

\begin{packeditemize}
    \item The $\mathsf{escrow}$ function attempts to perform an escrow operation on a single object $o$ within the transaction $tx$. It first calculates the $value$ by applying the operation $o.op$ on $o$'s current state. 
If the resulting $value$ satisfies the condition $o.con$, the function updates $o.value$ and appends this escrow request to the $elog$.

    \item The $\mathsf{allEscrowed}$ function is used to check whether owned objects in $tx$ with the decremental operation have been successfully escrowed. It initializes $allEscrowed(tx)$ as $true$, then iterates over each object $o$ in $tx$. If an object with a decremental operation and owned type has not been escrowed, $allEscrowed(tx)$ is set to $false$, and the function breaks out of the loop early. 

    \item The $\mathsf{commitEscrow}$ function is called to remove all the escrow requests of $tx$ from the $elog$. 

    \item The $\mathsf{abortEscrow}$ is called to undo and remove all the escrow requests of $tx$ in the $elog$.
\end{packeditemize}

The replica monitors the $plog$s and $glog$ to identify the next transaction ready for execution, referred to as the first pending transaction in the log. The first pending transaction is defined as the transaction for which all preceding transactions in the log, with smaller indexes, have been confirmed. 

\bheading{Execute transactions in $plog$.}
Upon identifying the first pending transaction $tx$ in \( plog[i] \), for each object \( o \) in \( tx \), if \( o \) is an owned object belongs to block $b$ in the current instance with a decremental operation, the replica attempts to perform an $\mathsf{escrow}$ on \( o \) for \( tx \).  The escrow is performed on the system state $b.S$ referred to by the transaction or any subsequent state derived from it through valid updates.
If the escrow operation on \( o \) fails, all escrows in \( tx \) are aborted. Then the transaction \( tx \) is aborted, which means $tx$ is removed from all the partial logs $plog$ and the global log $glog$. If all required escrows for \( tx \) are successful and \( tx \) is a payment transaction, the system commits all escrows in \( tx \), and proceeds to other objects in \( tx \) with incremental operations by applying the incremental operation to update \( o.value \) accordingly.

\bheading{Execute transactions in $glog$.}  When a replica identifies the first pending transaction \( tx \) in \( glog \), it first checks if this is the last occurrence of \( tx \) in \( glog \), as a transaction may appear in multiple positions within \( glog \). If it is not the last occurrence, \( tx \) is simply removed from the current position in \( glog \). If it is the last occurrence, the replica proceeds to execute \( tx \) and verify whether all escrow requests within the transaction are successful. If both the execution and escrow checks are successful, the replica commits all escrow requests associated with \( tx \). If any escrow request fails, all escrows in \( tx \) are aborted, and \( tx \) is removed from \( glog \) in all its occurrences. Notably, executing \( tx \) in the global log (\( glog \)) must strictly align with the global state at its designated position in $glog$.This ensures that \( tx \) is processed consistently across all replicas under the same system state, thereby maintaining overall consistency. 

\subsection{Checkpoint and Garbage Collection}\label{sec:checkpoint}
\sysname runs a simple checkpoint protocol at the end of each epoch. Upon epoch completion, each replica broadcasts a checkpoint message to all other replicas, containing a signed digest summarizing the blocks it processed during that epoch. When a replica receives a sufficient quorum of matching checkpoint messages (\ie, from at least $2f+1$ replicas), it creates a stable checkpoint, which serves as a verified snapshot of the epoch’s blocks. This stable checkpoint enables the replica to securely discard the data from the completed epoch, allowing for efficient garbage collection. Additionally, unexecuted transactions can be discarded, preventing the retention of transactions that will never be executed.

\section{Correctness Analysis} \label{sec:analysis}
In this section, we prove the safety and liveness properties of \sysname. For safety, we show that replicas in the same state execute an identical set of transactions, with payment transactions yielding consistent object values across replicas due to their commutative nature, and contract transactions producing consistent results by being executed sequentially under a global ordering. For liveness, we demonstrate that any transaction received by an honest replica will eventually be proposed and delivered by its corresponding SB instances, with all instances achieving consistent execution outcomes (success or failure). As a result, each transaction will eventually be confirmed by all honest replicas.

\begin{lemma}\label{lemma:exeagree}
If two honest replicas 
reach the same state \( S \), they must successfully execute the same set of transactions.
\end{lemma}

\begin{proof}
Let \( r_1 \) and \( r_2 \) be two honest replicas in the same state \( S = (sn_0, sn_1, \ldots, sn_{m-1}) \), where \( sn_i \) represents the maximum sequence number of transactions delivered by SB instance \( i \). By the agreement property of the SB protocol, both replicas \( r_1 \) and \( r_2 \) must have delivered the same sequence of transactions for each SB instance up to \( sn_i \).  For contract transactions, which are globally ordered, both replicas execute them in the same sequence and under the same state, ensuring that the outcome is identical. {For a payment transaction in a block $b$, which refers to the state $b.S$ agreed upon by all the honest replicas, the execution is guaranteed to succeed as long as both replicas execute them under $b.S$ or any subsequent state derived from it through valid updates.}

\end{proof}

\begin{lemma}\label{lemma:ownedvalue}
If two honest replicas successfully execute the same set of transactions \( T \), they must have consistent final values for all owned objects involved in \( T \).
\end{lemma}

\begin{proof}
If two honest replicas \( r_1 \) and \( r_2 \) successfully execute the same set of transactions \(T\).
Let \(O = \{o_1, o_2, \ldots, o_n\}\) be the set of all owned objects involved in $T$,
each transaction \(tx_i \in T\) has the form:
   \[
   tx_i = \{(o_{i_1}, o_{j_1}, \Delta_{1}), (o_{i_2}, o_{j_2}, \Delta_{2}), \ldots, (o_{i_k}, o_{j_k}, \Delta_{k})\},
   \]
where each tuple \((o_{i_p}, o_{j_p}, \Delta_p)\) denotes a decremental operation on object \(o_{i_p}\) and an incremental operation on object \(o_{j_p}\) by the same amount \(\Delta_p\).
The final value of any object \(o \in O\) after applying all transactions in \(T\) is:
  \[
value(o) = value_0(o) + \underbrace{\sum_{\substack{tx_i \in T, \\ (o, o_j, \Delta) \in tx_i}} \Delta}_{o~is~payee} - \underbrace{\sum_{\substack{tx_i \in T, \\ (o_i, o, \Delta) \in tx_i}} \Delta}_{o~is~payer}.
\]
Here \(value_0(o)\) is the initial value of $o$,  the first sum accumulates all increments to \(o\), and the second sum accumulates all decrements from \(o\). 
Therefore, for any permutation \( L \) of \( T \) that ensures the successful execution of all transactions, executing in the order specified by \( L \) will result in the same final values for all objects in \( O \). 
\end{proof}

\begin{lemma}\label{lemma:sharedvalue}
If two honest replicas successfully execute the same set of transactions \( T \), they must have consistent final values for all shared objects involved in \( T \).
\end{lemma}

\begin{proof}
Since the execution order of contract transactions in \( T \) is the same for both replicas, we can consider the sequence \( tx_1, tx_2, \dots, tx_n \) representing this fixed order.
For each transaction \( tx_i \) in the sequence, all replicas apply \( tx_i \) at the same step in the execution.
Let the value of any shared object \( o \) be \( value_i(o) \) after transaction \( tx_i \) is executed on it.
Because \( tx_i \) is deterministic and is executed at the same step on all replicas, any change to \(value_i(o) \) will be identical across replicas.
Thus, the final values for all shared objects involved in \( T \) will be consistent across replicas \( r_1 \) and \( r_2 \), provided they started with consistent initial values.
\end{proof}





\begin{theorem}[Safety]    
If two honest replicas reach the same state \( S \), they must have consistent values for all objects. 
\end{theorem}

\begin{proof}
Given that \( r_1 \) and \( r_2 \) are in the same state \( S \) (\( S_{r_1} = S_{r_2} \)), Lemma~\ref{lemma:exeagree} ensures that both replicas successfully execute the same set of transactions.
By Lemma~\ref{lemma:ownedvalue} and Lemma~\ref{lemma:sharedvalue}, if the replicas have consistent initial values for all objects (owned and shared) involved in \( T \), they must have consistent final values for those objects after execution.
Thus, for all objects \( o \), \( value_{r_1}(o) = value_{r_2}(o) \).
\end{proof}

\begin{lemma}\label{eventually1}
For any transaction $tx$ sent by a correct client $c$, $tx$ is eventually broadcast in all SB instances it belongs to.
\end{lemma}

\begin{proof}
For any transaction \( tx \) sent by a correct client \( c \), the transaction is sent to at least \( f + 1 \) replicas, ensuring that at least one honest replica receives it and adds it to the relevant buckets. For each instance it belongs to, if the leader is honest, it will broadcast \( tx \). If the leader is Byzantine, it will be detected and a new honest leader will be selected eventually. This new leader will broadcast \( tx \) from the bucket, ensuring that \( tx \) is eventually broadcast.
\end{proof}

\begin{lemma}[Atomicity]\label{lemma:atomicity}
For a payment transaction $tx$, either all its sub-transactions are executed successfully, or aborted.
\end{lemma}

\begin{proof}
By the Lemma~\ref{eventually1} and the termination property of the SB protocol, all sub-transactions of $tx$ will eventually be delivered. Each sub-transaction is then processed using the escrow mechanism. According to this mechanism, if any sub-transaction fails during the escrow phase, all associated escrow requests of $tx$ are aborted, meaning that all sub-transactions will be aborted. Conversely, if all escrow requests are successfully committed, all sub-transactions are executed successfully. Thus, the atomicity of the transaction is preserved. 
\end{proof}

\begin{lemma}\label{eventually2}
For any transaction $tx$ broadcast by an honest replica, $tx$ is eventually confirmed by all the honest replicas.
\end{lemma}
\begin{proof}
By the termination property of the SB protocol, a transaction broadcast by an honest replica will eventually be delivered by all honest replicas. 
For a payment transaction, all the sub-transactions will be executed or aborted after being delivered by the SB instances (Lemma~\ref{lemma:atomicity}). Thus, the payment transaction will be confirmed.
Contract transactions are processed by the global ordering algorithm after being delivered by the SB instances. This algorithm traverses all delivered transactions to ensure they are globally ordered and confirmed. Consequently, all honest replicas will eventually confirm all delivered contract transactions.
\end{proof}

\begin{theorem}[Liveness]
Given a transaction \( tx \) from a correct client \( c \), $tx$ is eventually confirmed by all honest replicas.
\end{theorem}
\begin{proof} 
Assuming the transaction \( tx \) is sent by a correct client, by Lemma~\ref{eventually1}, \( tx \) will be broadcast in all SB instances it belongs to. Following this, according to Lemma~\ref{eventually2}, \( tx \) will be confirmed by all honest replicas.
\end{proof}

\section{Performance Evaluation} \label{sec:evaluation} 
In this section, we evaluate the performance of \sysname in different scenarios and compare \sysname against ISS~\cite{stathakopoulou2022state}, RCC~\cite{gupta2021rcc}, Mir-BFT~\cite{MIR-BFT}, DQBFT~\cite{dqbft}, and Ladon~\cite{Ladon2025}, the state-of-the-art Multi-BFT protocols. 
We implemented \sysname in Golang~\cite{golang}.
{ The evaluation results are illustrated using ChiPlot\footnote{https://www.chiplot.online/}. }
We use PBFT~\cite{castro2002practical} consensus protocol to implement SB instances.
With our experiments, we answer the following questions:
\begin{packeditemize}
    \item \textbf{Q1:} How does \sysname perform as compared with ISS, RCC, Mir-BFT, DQBFT, and Ladon with a varying number of replicas? (\secref{sec:expperformance1})   
    
    \item \textbf{Q2:}How does \sysname perform with different proportions of payment transactions? (\secref{sec:proportions})  

    \item \textbf{Q3:} What is the latency breakdown of \sysname as compared to ISS? (\secref{sec:breakdown})

    \item \textbf{Q4:} {How does \sysname perform under faults?  (\secref{sec:faults})}
\end{packeditemize}


\subsection{Implementation and Experimental Setup} \label{sec:expset}

\bheading{Implementation.}
\sysname\footnote{The source code is available at \url{https://github.com/Hanzheng2021/Orthrus/}.} is implemented in Golang and built on top of the ISS\footnote{The source code for ISS is available at \url{https://github.com/hyperledger-labs/mirbft/tree/research-iss}.}~\cite{stathakopoulou2022state} platform. Golang was chosen for its concurrency features and performance, and the ISS platform provides a solid foundation for its infrastructure. 
Our implementation is fully integrated with the ISS platform, ensuring compatibility and allowing for direct comparisons with other protocols developed on the same foundation.

\begin{figure}[t]
\vspace{-2mm}
    \centering
    \begin{subfigure}[t]{0.24\textwidth}
        \centering
        \includegraphics[width=\textwidth]{ 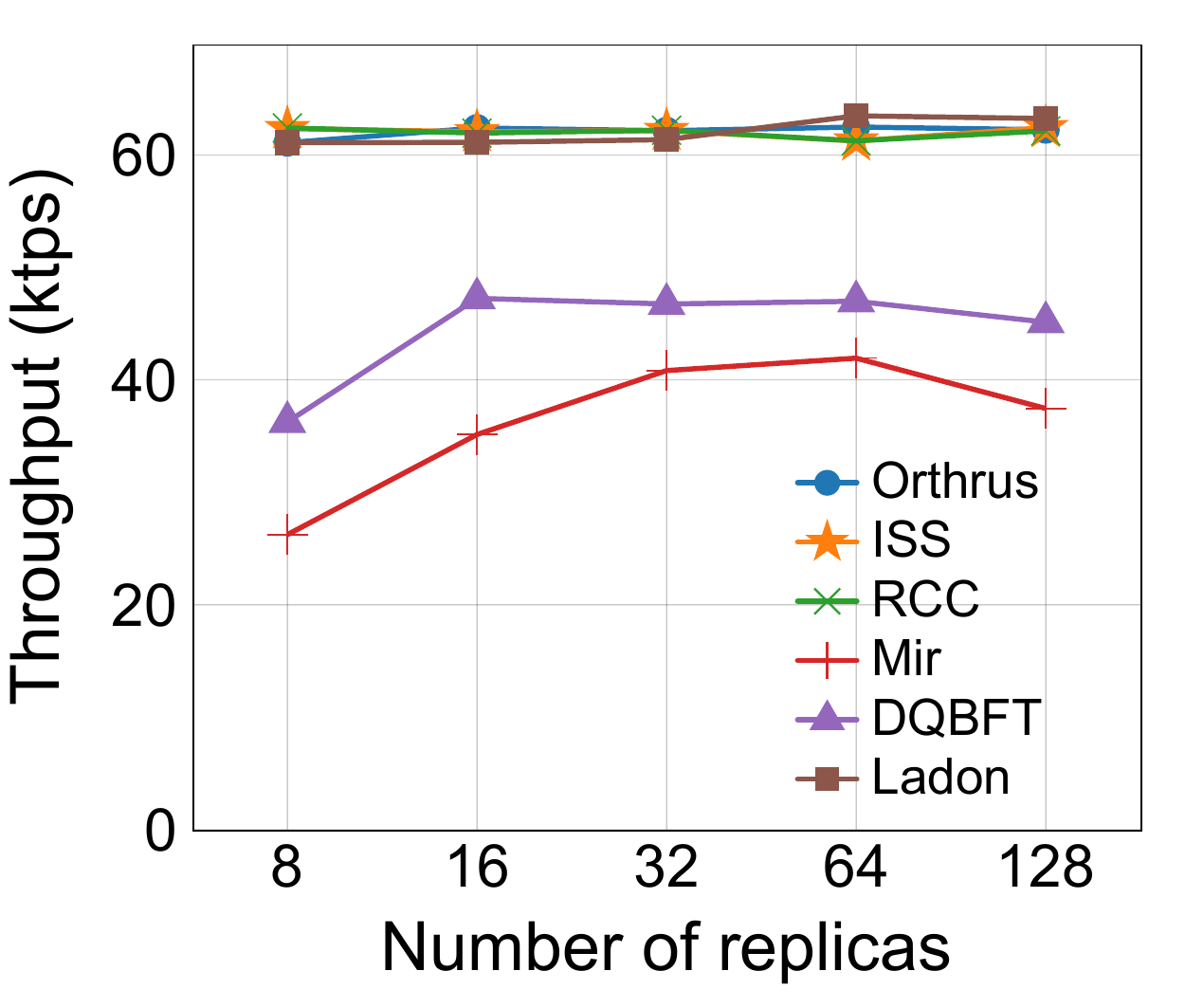}
        \caption{$\#$Stragglers = 0, WAN}
        \label{fig:wan1}
    \end{subfigure}
    \hfill
    \begin{subfigure}[t]{0.24\textwidth}
        \centering
        \includegraphics[width=\textwidth]{ 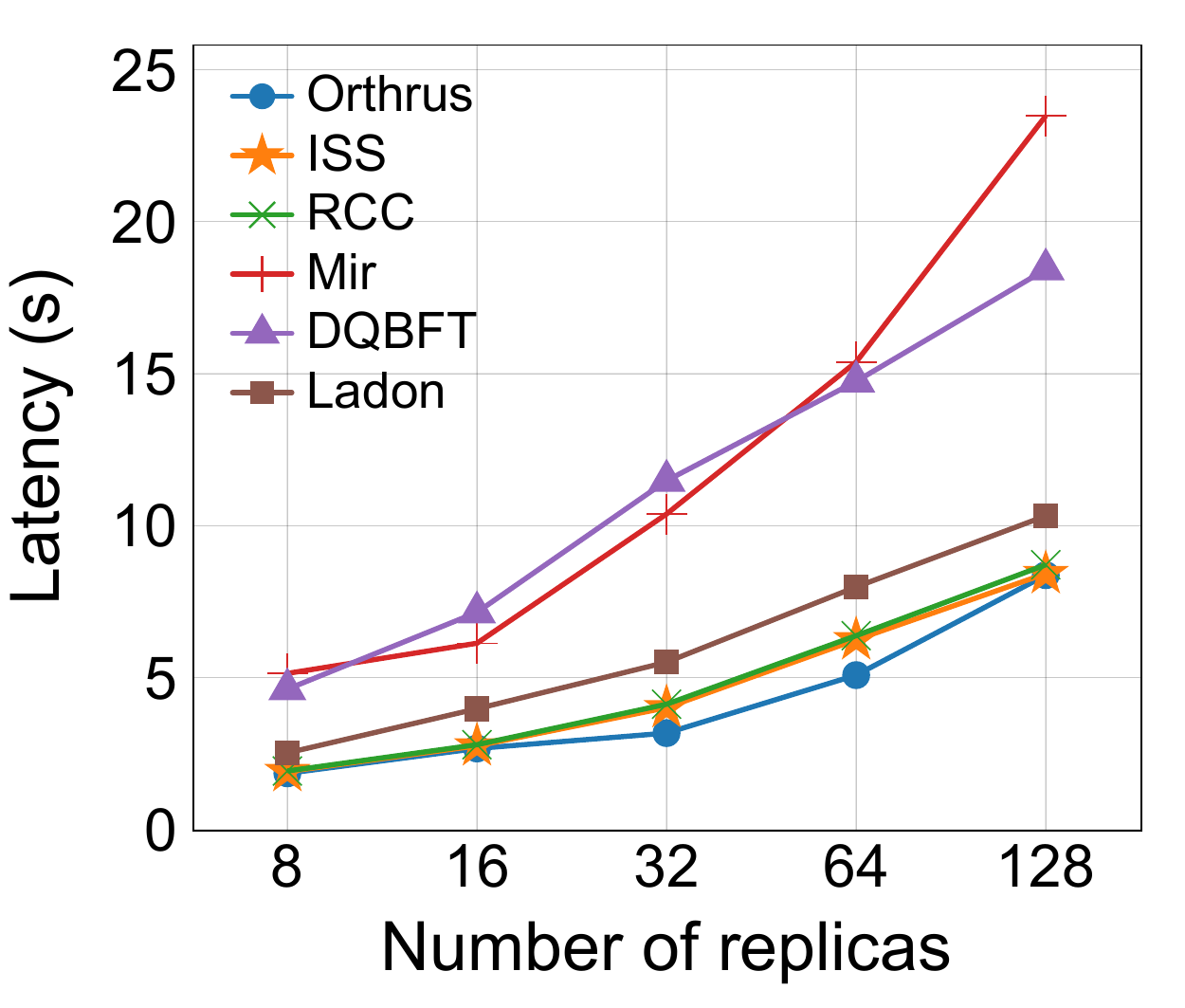}
        \caption{$\#$Stragglers = 0, WAN}
        \label{fig:wan3}
    \end{subfigure}
    \hfill\\
    \begin{subfigure}[t]{0.24\textwidth}
        \centering
        \includegraphics[width=\textwidth]{ 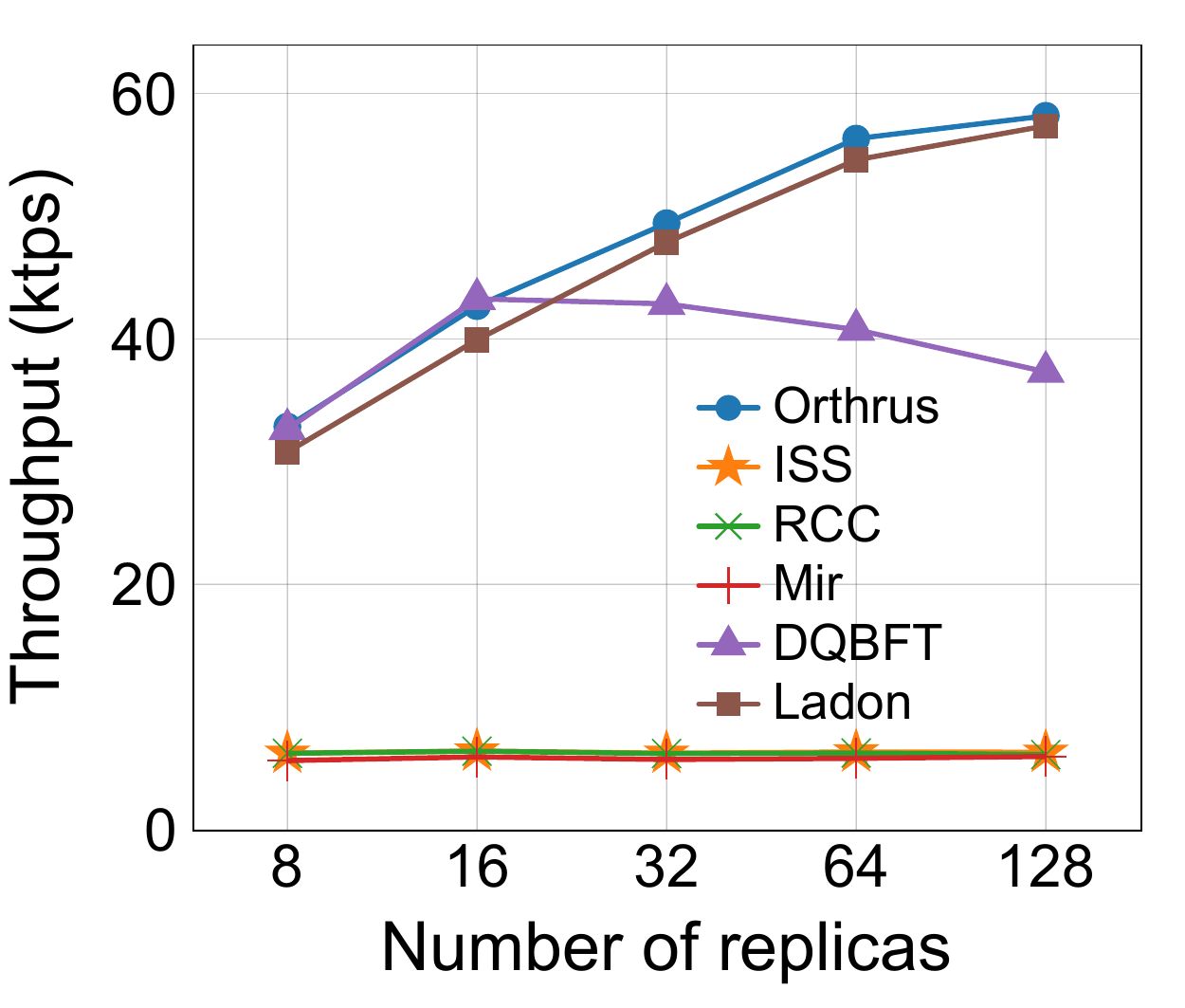}
        \caption{$\#$Straggler = 1, WAN}
        \label{fig:wan2}
    \end{subfigure}
    \hfill
    \begin{subfigure}[t]{0.24\textwidth}
        \centering
        \includegraphics[width=\textwidth]{ 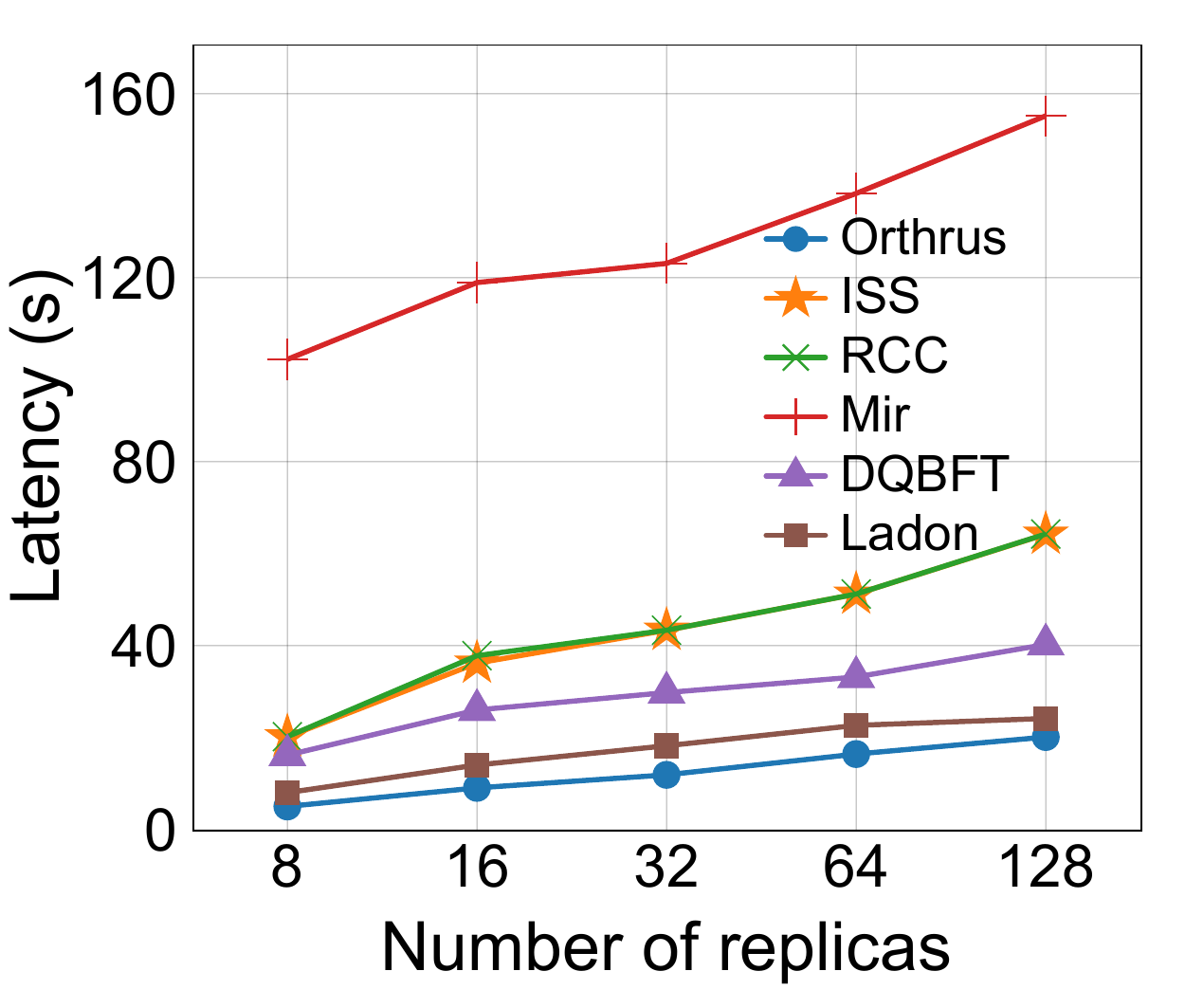}
        \caption{$\#$Straggler = 1, WAN}
        \label{fig:wan4}
    \end{subfigure}
    \hfill

    \caption{\textbf{Throughput and latency of \sysname, ISS, RCC, Mir, DQBFT, and Ladon  in WAN.}}
    \label{fig:performance}
\end{figure}

\bheading{Experimental setup.}
We deploy our protocols on AWS EC2 instances (c5a.2xlarge), with each instance representing a replica in our distributed system. Each instance is equipped with 8vCPUs, 16GB of RAM, and runs Ubuntu Linux 22.04. The experiments are conducted in both LAN and WAN environments. In the LAN setup, machines communicate over private network interfaces with a bandwidth of 1Gbps. For the WAN setup, instances are distributed across four Amazon cloud data centers located in France, the United States, Australia, and Tokyo, with network interfaces limited to 1Gbps. We utilize NTP for clock synchronization.

The dataset for our experiments is derived from the Ethereum blockchain, consisting of approximately 200,000 transactions extracted from blocks with heights ranging from 17,198,000 to 17,202,000. These transactions are drawn from a pool of 18,000 active accounts and include both payment and contract transactions, with payment transactions comprising 46\% of the transactions. This distribution is preserved in our experimental setup, using real transaction data as input. Once these transactions are processed, we reset the account states by introducing specific transactions and then re-execute the same set of 200,000 transactions to evaluate system performance under repeated workloads.

Each replica can function both as a leader for one instance and as a backup for others, i.e., $m = n$. To maximize throughput, we allow a large batch size of 4096 transactions, each carrying a payload of 500 bytes. We evaluate system performance under two different network conditions: with and without stragglers. In the straggler scenario, one instance operates at 10 times slower speed than the others. All experiments are repeated five times and the average results are reported.

\subsection{Throughput and Latency}\label{sec:expperformance1}
We evaluate the performance of \sysname, ISS, RCC, Mir, DQBFT, and Ladon without stragglers and with one straggler in both We evaluate two performance metrics: 1) throughput: the number of transactions {responded} to clients per second, and 2) latency: the average end-to-end delay from the moment clients submit transactions until they receive $f+1$ responses. We measure the peak throughput in kilo-transactions per second (ktps) before reaching saturation, along with the associated latency in seconds (s).

\begin{figure}[t]
\vspace{-4mm}
    \centering

    \begin{subfigure}[t]{0.24\textwidth}
        \centering
        \includegraphics[width=\textwidth]{ 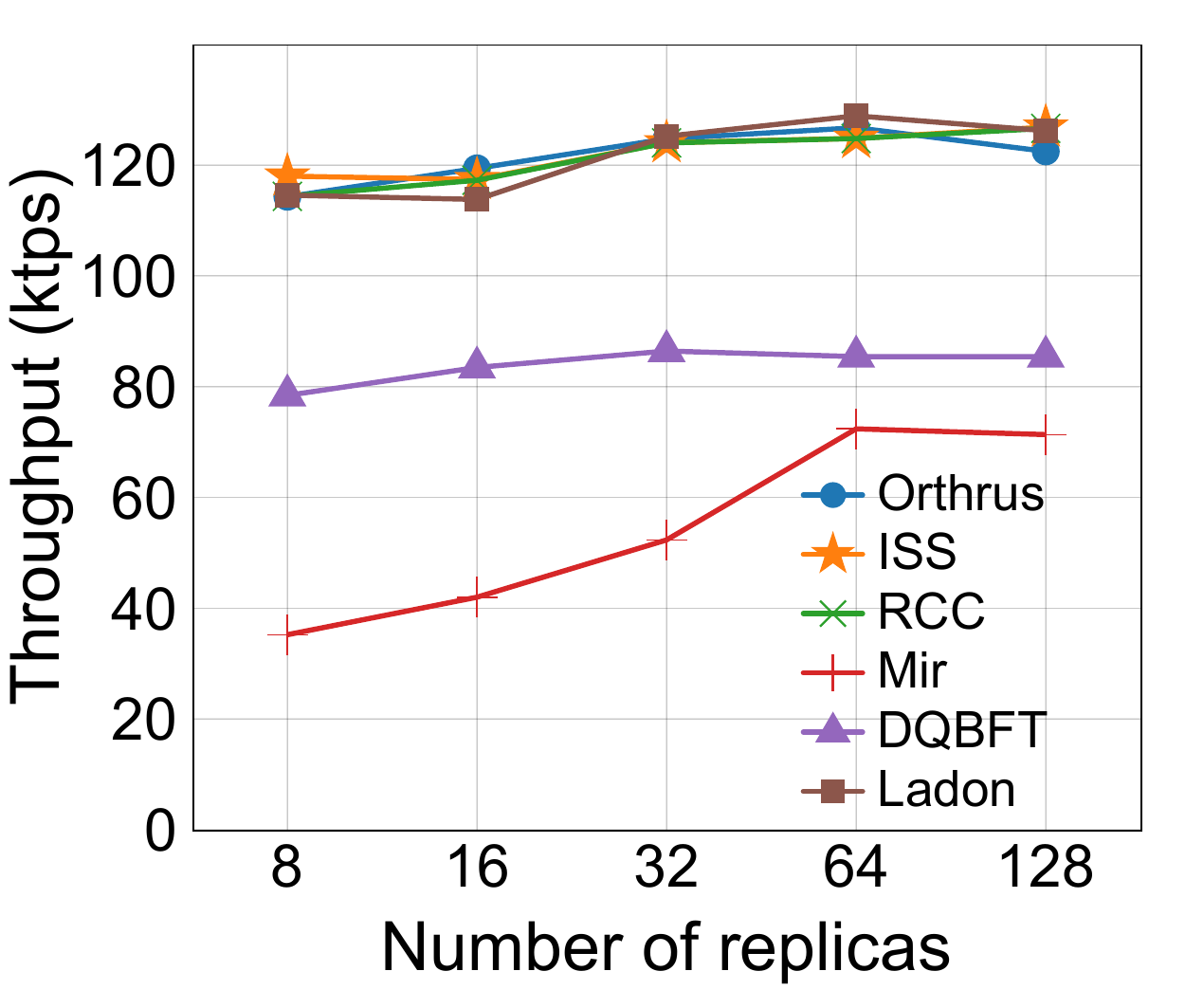}
        \caption{$\#$Stragglers = 0, LAN}
        \label{fig:lan1}
    \end{subfigure}
    \hfill
    \begin{subfigure}[t]{0.24\textwidth}
        \centering
        \includegraphics[width=\textwidth]{ 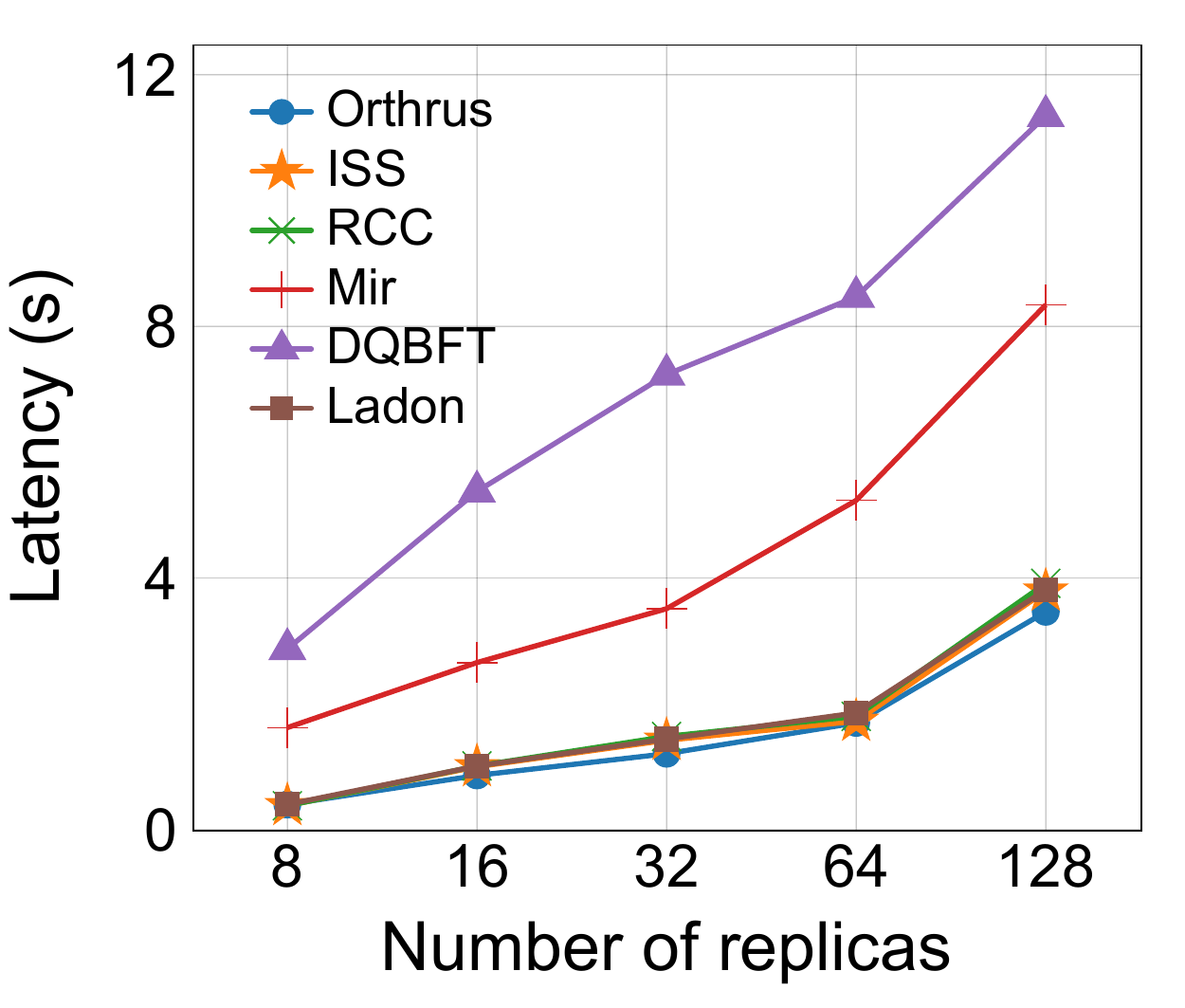}
        \caption{$\#$Stragglers = 0, LAN}
        \label{fig:lan3}
    \end{subfigure}
    \hfill\\
    \begin{subfigure}[t]{0.24\textwidth}
        \centering
        \includegraphics[width=\textwidth]{ 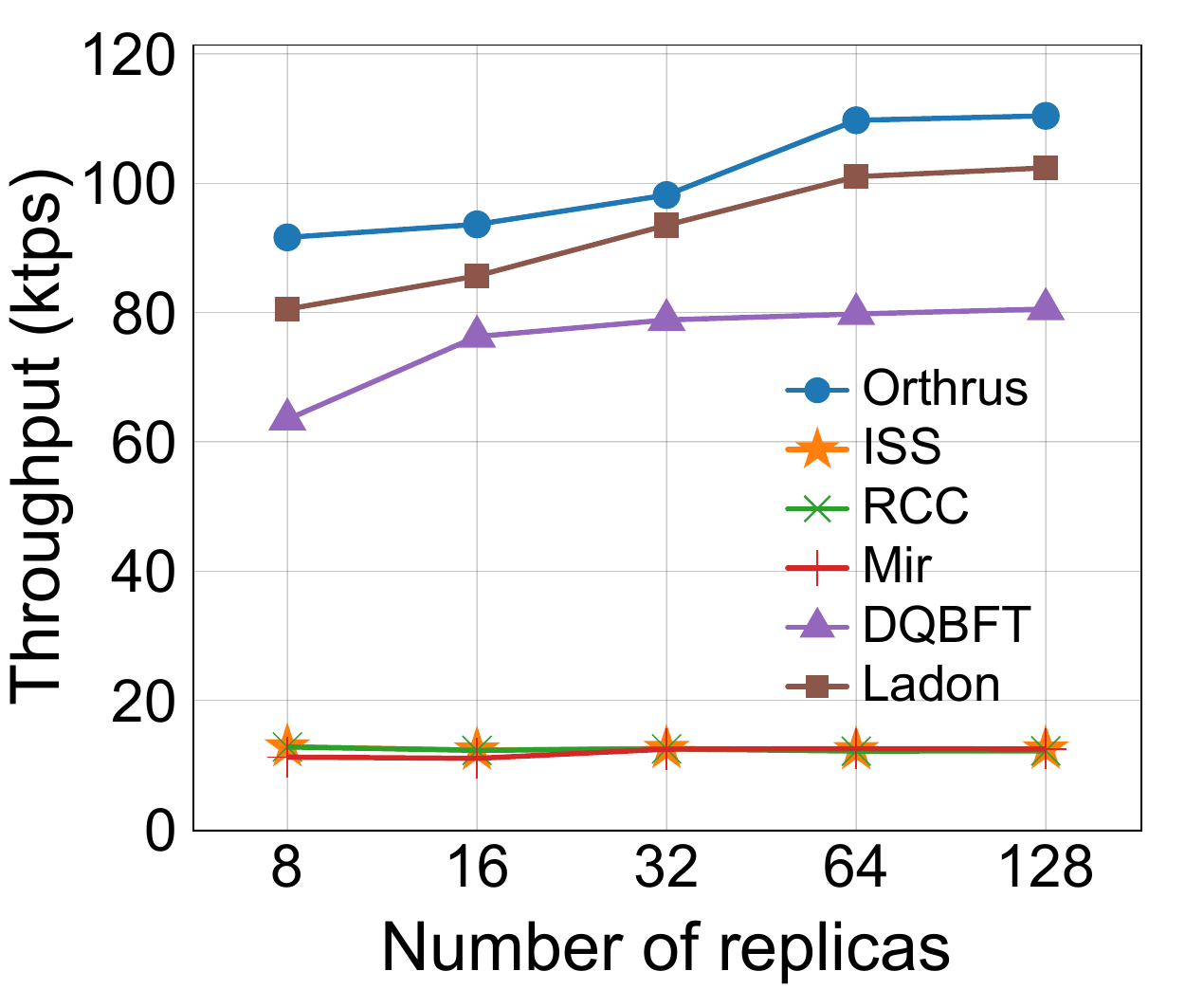}
        \caption{$\#$Straggler = 1, LAN}
        \label{fig:lan2}
    \end{subfigure}
    \hfill
    \begin{subfigure}[t]{0.24\textwidth}
        \centering
        \includegraphics[width=\textwidth]{ 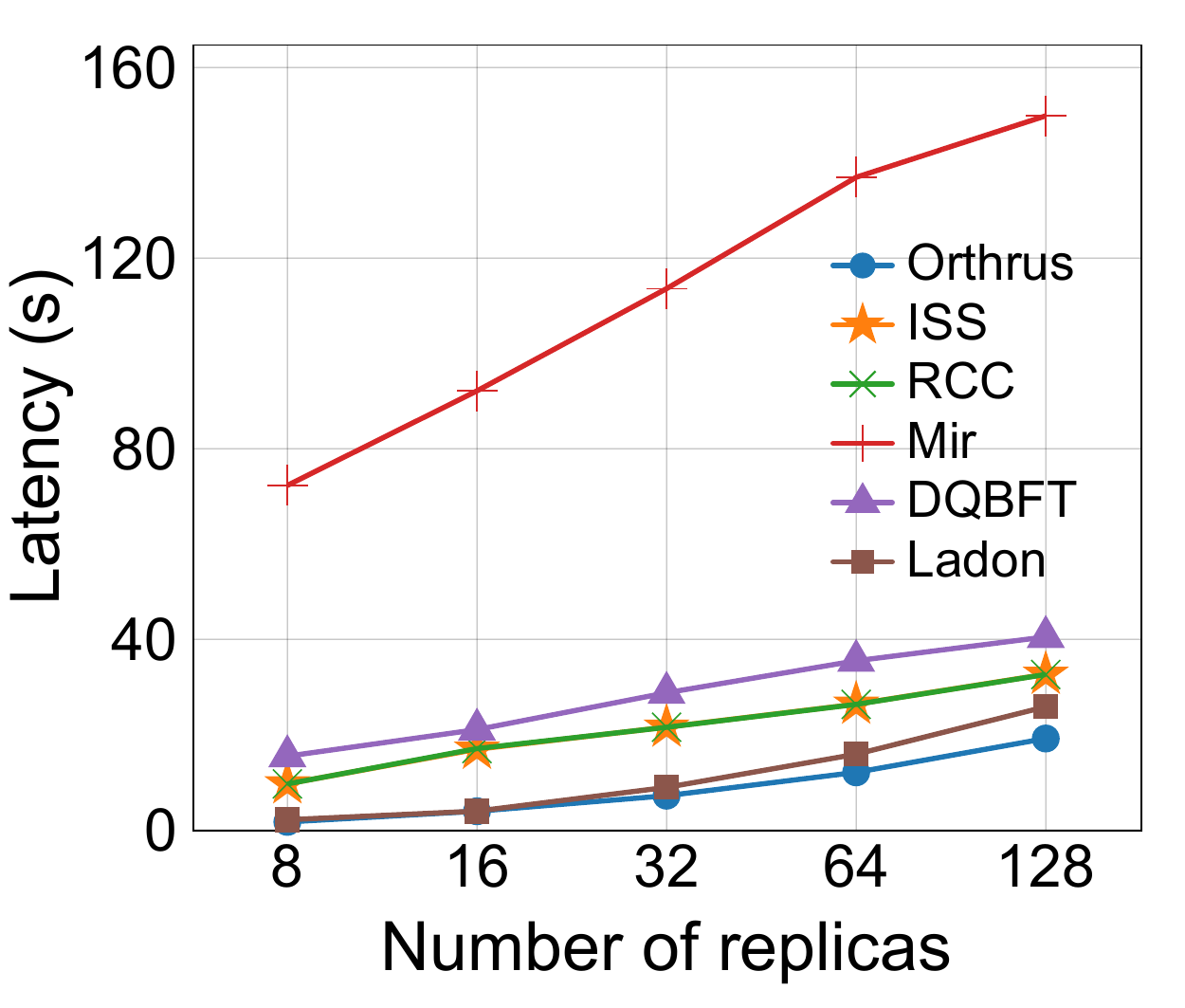}
        \caption{$\#$Straggler = 1, LAN}
        \label{fig:lan4}
    \end{subfigure}

    \caption{\textbf{Throughput and latency of \sysname, ISS, RCC, Mir, DQBFT, and Ladon {in} LAN.}}
    \label{fig:performancelan}
\end{figure}
\iheading{1) Performance in WAN.} 
\figref{fig:wan1} and \figref{fig:wan2} demonstrate that \sysname consistently maintains throughput within the top tier.
Without stragglers, \sysname shows very similar throughput with ISS and RCC.
{With one straggler, \sysname achieves comparable or superior throughput compared to dynamic global ordering Multi-BFT protocols.
Additionally, \sysname significantly surpasses pre-determined Multi-BFT protocols (\eg, ISS, RCC, Mir), delivering 9.5$\times$ higher throughput under the same conditions.} 
Comparing \figref{fig:wan1} and \figref{fig:wan2}, we observe that a straggler significantly impacts pre-determined Multi-BFT protocols, causing a drop of 89.9\% in throughput compared to those without stragglers on 128 replicas. \sysname and Ladon were less affected by stragglers,  experiencing only a $6.5\%$ and $9.2\%$ drop in throughput, respectively, compared to their performance without stragglers on 128 replicas.

\figref{fig:wan3} and \figref{fig:wan4} illustrate that \sysname consistently achieves the lowest latency compared to other protocols. Without stragglers, \sysname reduces the latency of 18.6\% of ISS and RCC with 64 replicas, 18.9\% of Ladon, 54.5\% and 64.3\% of DQBFT and Mir, on 128 replicas, respectively. With a straggler, the effect becomes more pronounced. On 128 replicas, \sysname reduces the latency of 68.6\% of ISS and RCC, 16.7\% of Ladon, 50.0\%, and 87.0\% of DQBFT and Mir, respectively.
By comparing \figref{fig:wan3} and \figref{fig:wan4}, we observe a similar impact of stragglers on latency, consistent with their effect on the throughput across different protocols.

\iheading{2) Performance in LAN.} \figref{fig:performancelan} shows that Protocols show similar trends in LAN with WAN, with higher throughput and lower latency. Without stragglers and with a straggler, \sysname always shows high throughput and low latency among all protocols. 
With one straggler, \sysname demonstrates approximately 8$\times$ higher throughput than ISS, RCC, and Mir, and 37.1\% and 7.9\% higher than DQBFT and Ladon, respectively, on 128 replicas. The latency is reduced by 16.7\%, and 50.0\% compared to Ladon and DQBFT, respectively, on 128 replicas.

The excellent performance of \sysname is mainly due to its fast confirmation mechanism for simple payment transactions, which allows some transactions to avoid global ordering and thus reduces latency. In the absence of stragglers, the delay introduced by global ordering is minimal, so \sysname shows only a slight advantage. However, in the presence of stragglers, the delay caused by global ordering is significant, giving \sysname a substantial advantage over pre-determined systems. 

\begin{figure}[t]
\vspace{-4mm}
	\centering
    \begin{subfigure}[t]{0.24\textwidth}
	\centering
        \includegraphics[width=\linewidth]{ 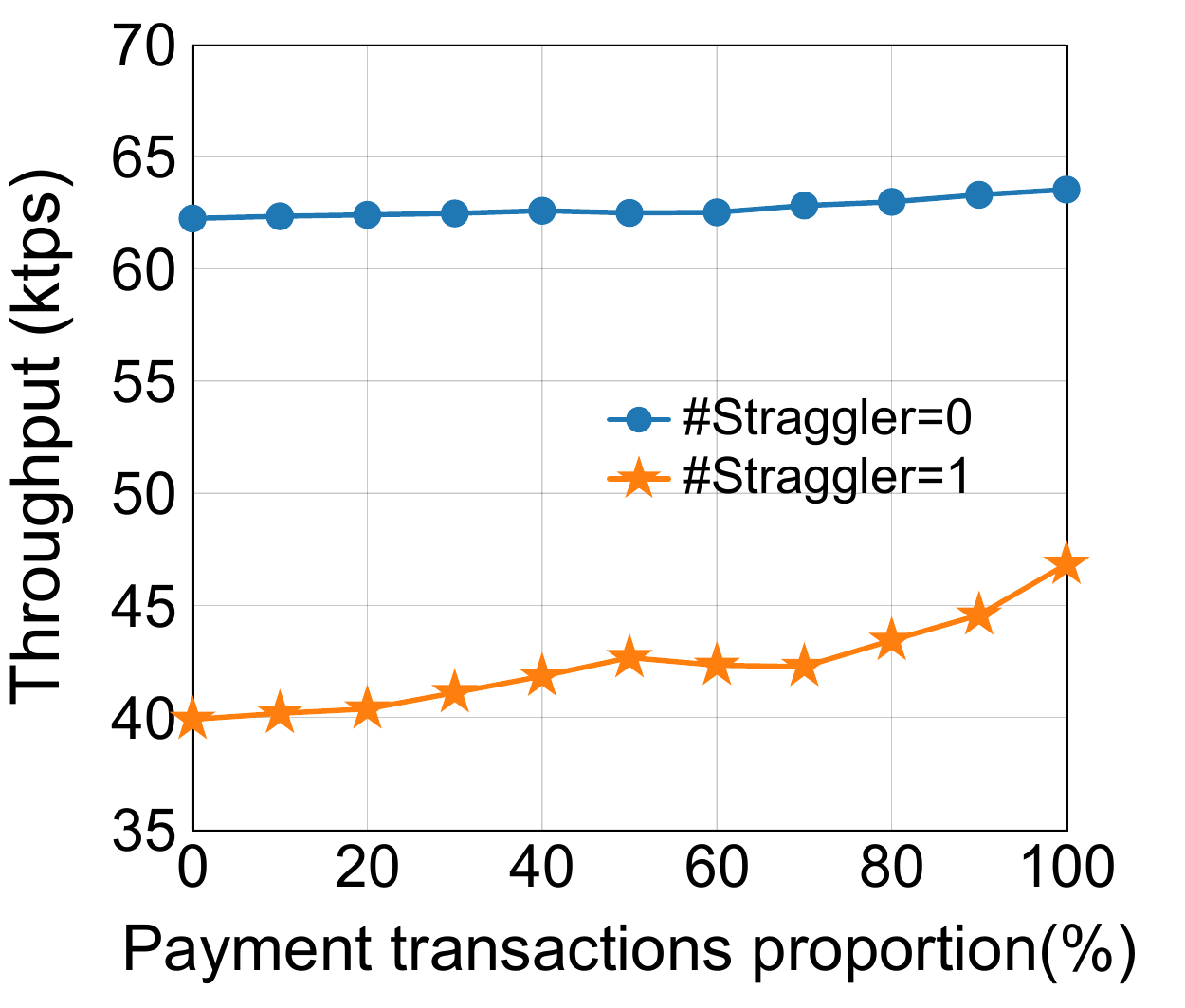}
        \caption{{Throughput}}
    \end{subfigure}
    \hfill
    \begin{subfigure}[t]{0.24\textwidth}
	\centering
        \includegraphics[width=\linewidth]{ 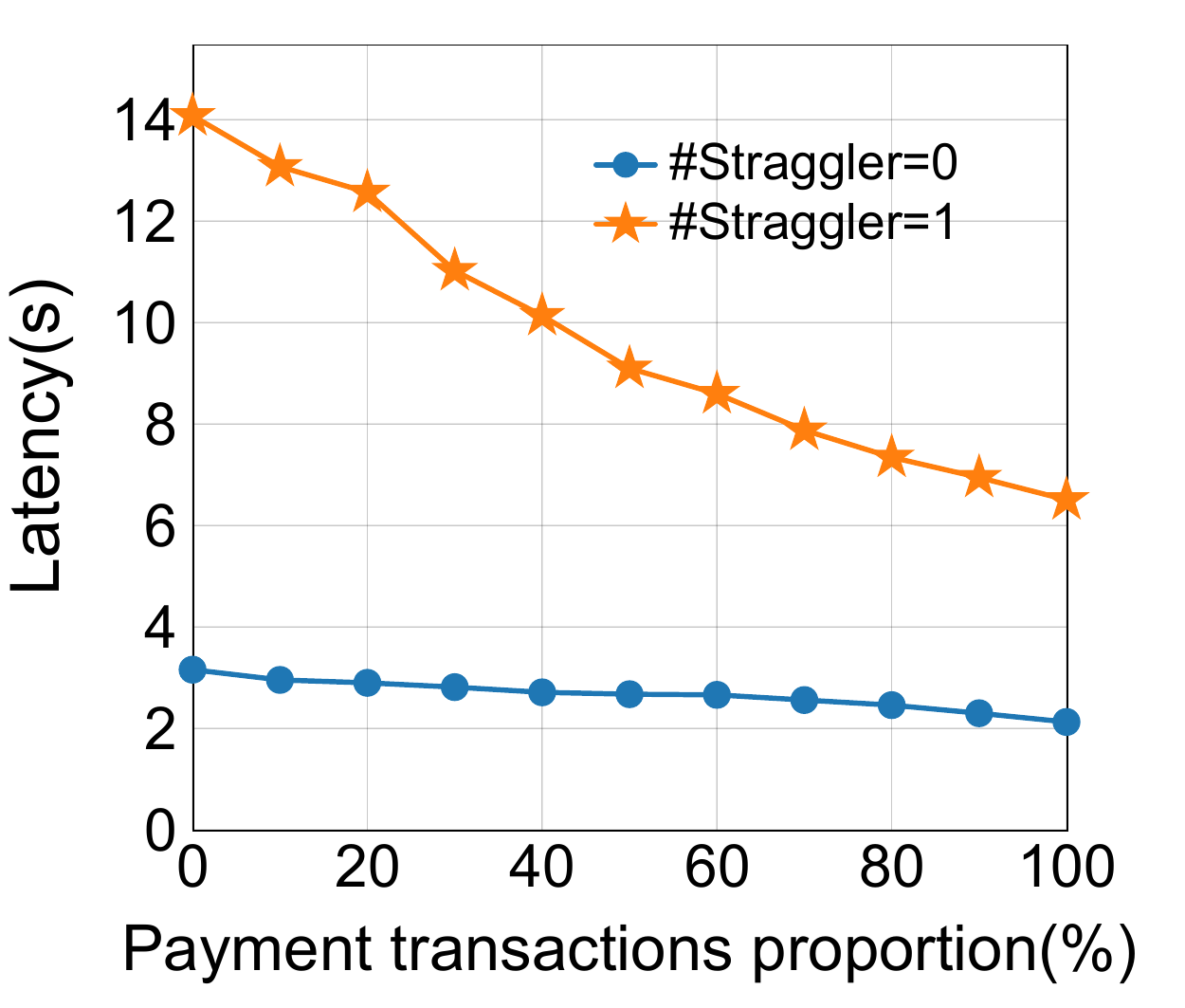}
        \caption{{Latency}}
    \end{subfigure}
    
    \caption{\textbf{Throughput and latency of \sysname under different payment transactions proportions in WAN.}}
	\label{fig:proportionTest}
\end{figure}

\subsection{Varying Payment Transactions Proportions}\label{sec:proportions}
\figref{fig:proportionTest} illustrates the performance of \sysname with varying proportions of payment transactions. The experiments are conducted on 16 replicas in WAN with/without a straggler. As depicted in \figref{fig:proportionTest}, the throughput increases and the latency decreases as the proportion of payment transactions rises. Specifically, without stragglers, when the proportion of payment transactions increases from 0\% to 100\%, the throughput increases by up to 2.1\% and the latency decreases by 32.7\%. However, with stragglers, the throughput sees a substantial increase of up to 17.3\%, while latency drops by 53.8\%. This pronounced change in the presence of stragglers is due to the higher proportion of payment transactions, allowing more transactions to be quickly confirmed, thereby greatly enhancing the overall processing efficiency.

\subsection{Latency Breakdown} \label{sec:breakdown}
We present a detailed latency breakdown of \sysname and ISS to understand the average spent time of different stages. We divide the process into five distinct stages: 1) Sending transactions: This stage spans the period from when the client sends the transaction until the replica receives it; 2) Pre-processing: This stage spans from receiving the transaction to broadcasting the transaction at the replica; 3) Partial Ordering: This stage extends from when the replica broadcasts a transaction to the delivery of the transaction; 4) Global Ordering: This begins with the delivery of the transaction and ends with the confirmation of the transaction;  and 5) Reply: This final stage covers when $f+1$ replicas confirm the transaction to when the client receives $f+1$ replies. 

The experiment was conducted in WAN with 16 replicas and one straggler. For each stage, the average time across all transactions is measured. 
The results are illustrated in \figref{fig:breakdown}. 
For the transaction sending and reply stages, the latency of \sysname and ISS are comparable, indicating similar communication infrastructure efficiencies. A similar trend is observed in the preprocessing and partial ordering stages.
One notable difference between the two protocols is evident in the global ordering stage, which exhibits a substantial delay in ISS, taking 34.5 seconds when a straggler is present and accounting for a remarkable 92.8\% of the total latency. 

This detailed breakdown demonstrates that the global ordering stage in ISS introduces a significant delay when stragglers are involved, becoming a bottleneck and exacerbating latency issues. In contrast, \sysname achieves significantly lower latency in the presence of stragglers, taking only 7.4 seconds for the global ordering. This lower latency is attributed to \sysname's efficient handling of transactions through its fast path mechanism, which minimizes the impact of slow global ordering. Additionally, the cost of global ordering can be further reduced as the proportion of payment transactions increases.

\begin{figure}[t]
        \includegraphics[width=0.5\textwidth]{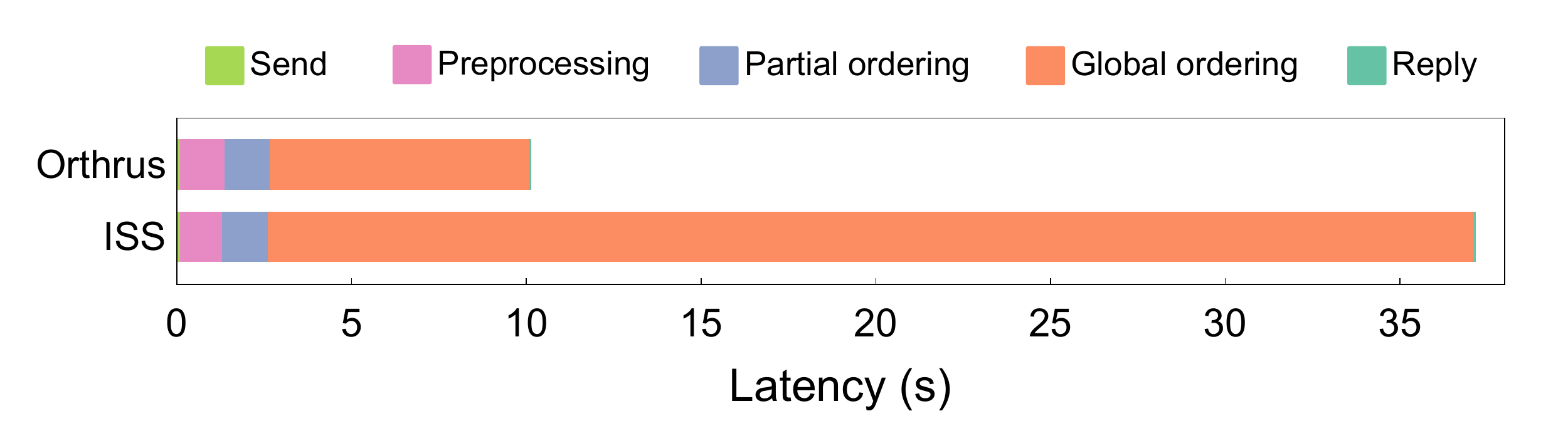}
    \caption{\textbf{Breakdown of latency in ISS versus \sysname.}}
    \label{fig:breakdown}
\end{figure}

\subsection{Performance Under Faults}\label{sec:faults}

In this section, we study the performance of \sysname under faults in a WAN of 16 replicas, allowing up to 5 faulty replicas. Our analysis covers both detectable faults and undetectable faults. We evaluate scenarios with 0, 1, and 5 faulty replicas. The PBFT view change timeout is set at $10$ seconds.

\begin{figure}[t]
\vspace{-4mm}
	\centering
    \begin{subfigure}[t]{0.5\textwidth}
	\centering
        \includegraphics[width=\linewidth, height=0.3\linewidth]{ 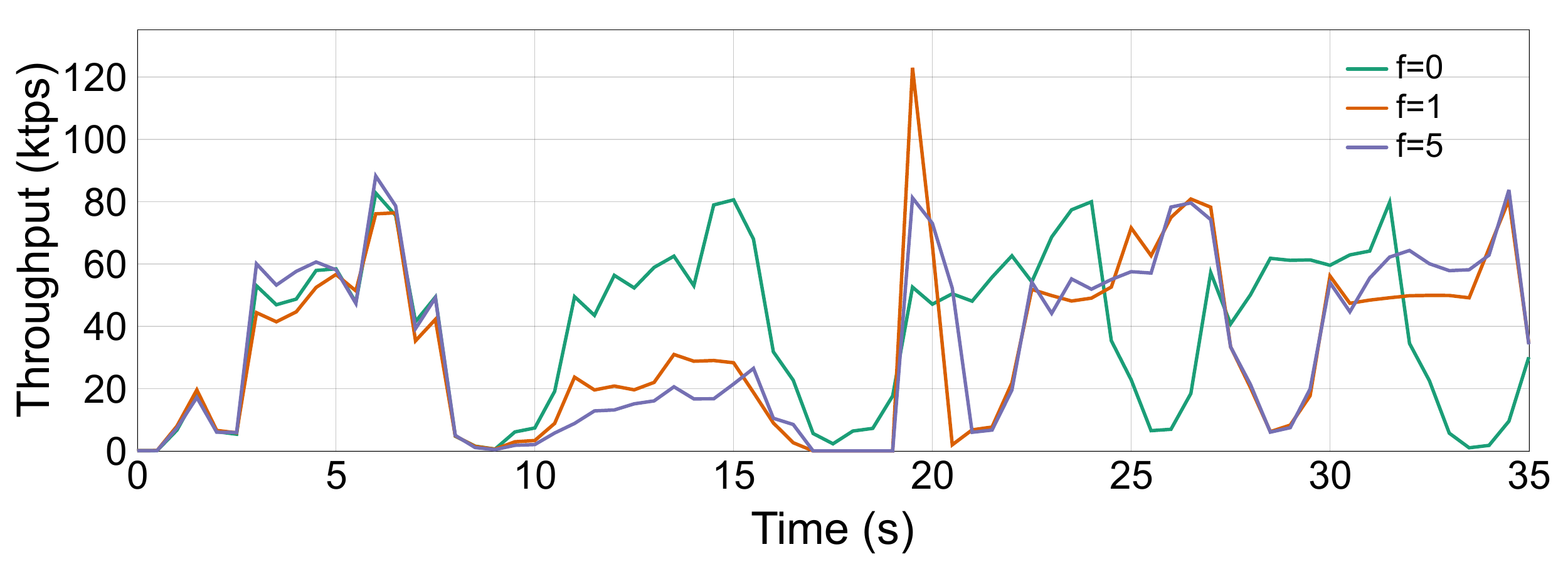}
        \caption{{Throughput average  (over 0.5s intervals) over time.}}
	\label{fig:detectablea}
    \end{subfigure}
    \hfill
    \begin{subfigure}[t]{0.5\textwidth}
	\centering
        \includegraphics[width=\linewidth, height=0.3\linewidth]{ 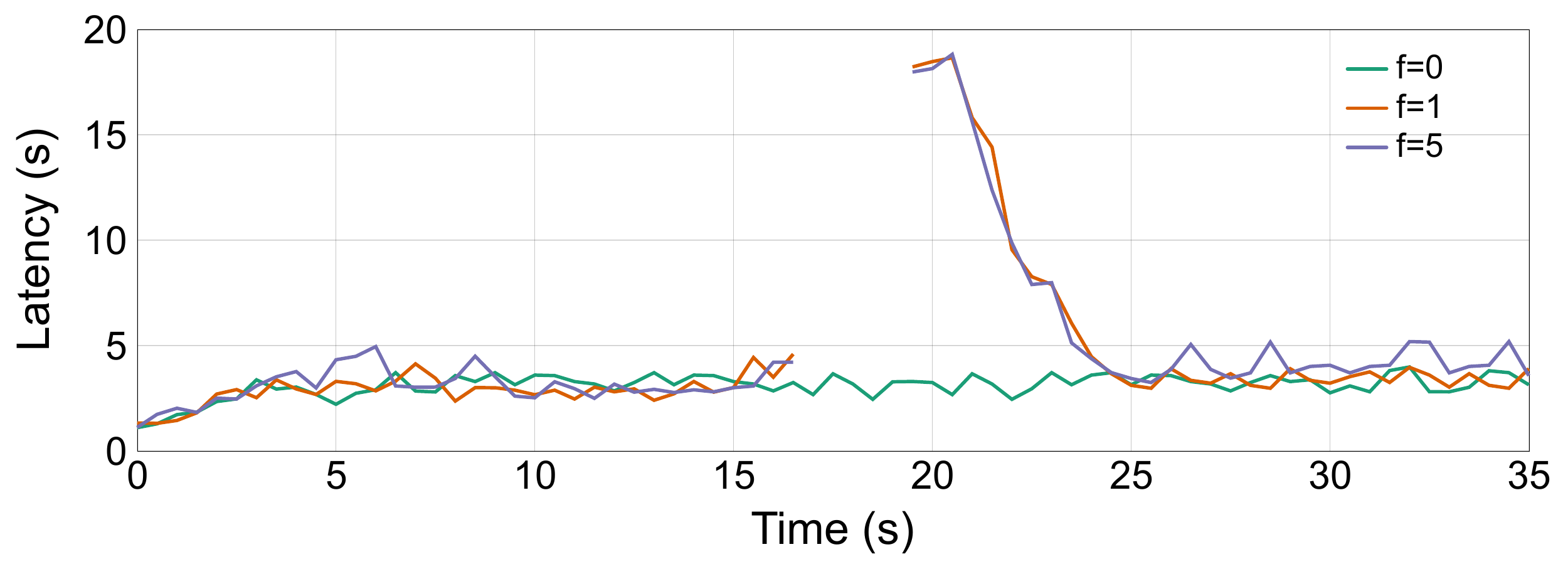}
        \caption{{Latency average  (over 0.5s intervals) over time.}}
        	\label{fig:detectableb}
    \end{subfigure}
    \caption{\textbf{{Throughput and latency average of \sysname over time with 0, 1, 5 faults. The faults occur at 9 seconds.}}}
	\label{fig:detectable}
\end{figure}

\bheading{Detectable faults.} We analyze how detectable faults impact the throughput and latency of \sysname.
\figref{fig:detectable} shows the throughput and latency average over time. In \figref{fig:detectablea}, the short drops to 0 in throughput correspond to the epoch change.
During the first epoch (0-9 seconds), no faults occur, and the system operates at its expected performance level. 

At the start of the second epoch (9 seconds), faults are introduced, causing a significant drop in throughput by more than 50\%. Specifically, when there is one faulty replica, the throughput drops by approximately 50\%, whereas, with 5 faulty replicas, the decrease is slightly more pronounced. 
This behavior is influenced by the transaction composition: 54\% of the transactions are contract-based and require global ordering, making them unconfirmed regardless of the number of faults.  In contrast, partially ordered payment transactions are only affected in instances led by faulty replicas, meaning their impact scales with the number of faults. During the decline in throughput, latency remains nearly unchanged. This is because transactions confirmed in this phase are from instances led by honest replicas, which are less affected by faulty replicas.

Upon fault detection, the system initiates the view change process to recover, completing at 19 seconds. Following this, throughput rises sharply due to the confirmation of transactions in the global log. In \figref{fig:detectableb}, latency exhibits a similar trend, as these confirmed transactions have been blocked for a significant duration, resulting in a high average latency. The brief interruption in the latency curve before 19 seconds corresponds to the period where throughput drops to zero in \figref{fig:detectablea}, indicating that no transactions were confirmed during this time.

By the beginning of the third epoch (22 seconds), throughput recovers as the system stabilizes,  and latency gradually stabilizes as well, demonstrating \sysname’s resilience to detectable faults.

\begin{figure}[t]
\vspace{-4mm}
	\centering
    \begin{subfigure}[t]{0.24\textwidth}
	\centering
        \includegraphics[width=\linewidth]{ 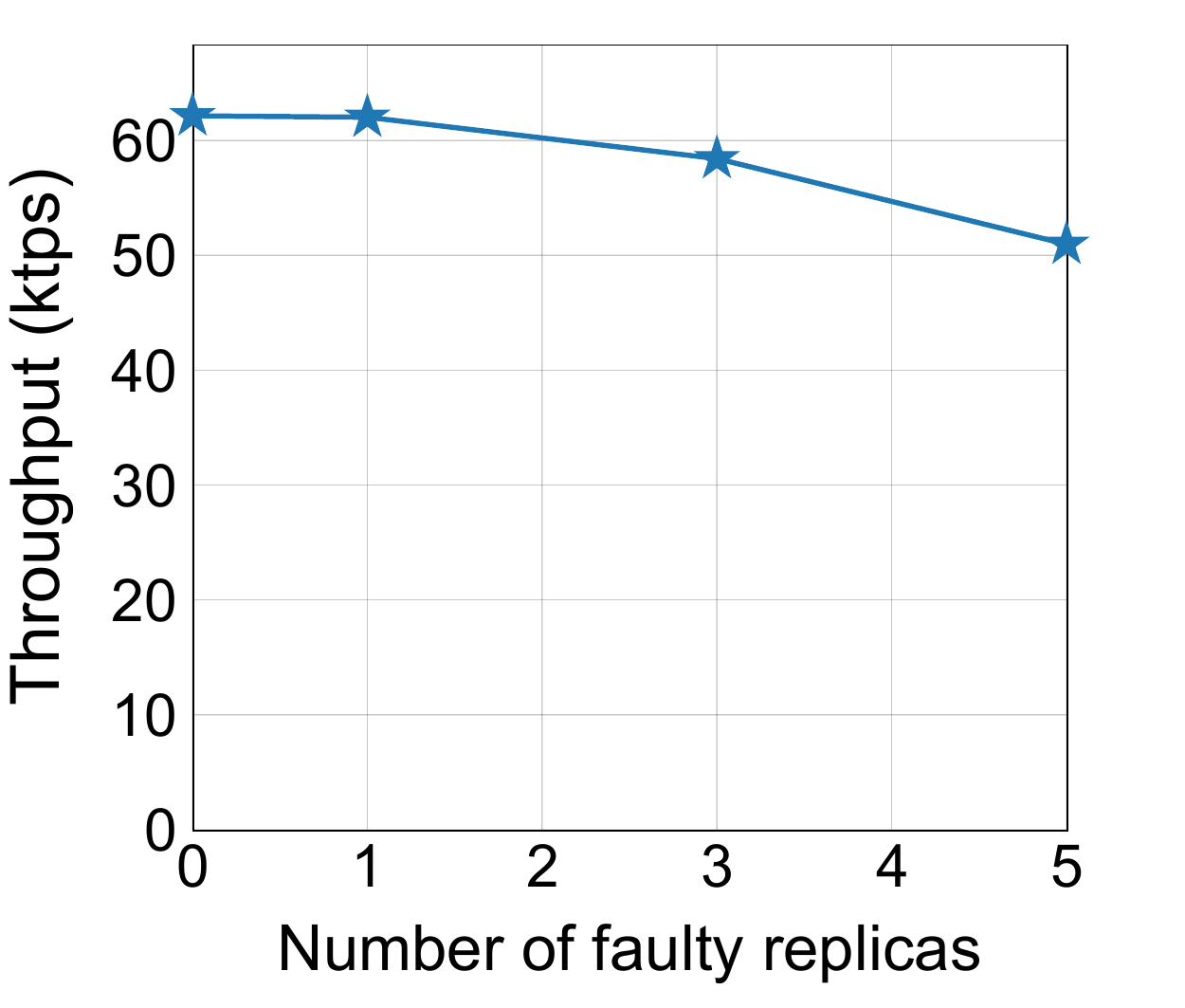}
        \caption{{Throughput}}
    \end{subfigure}
    \hfill
    \begin{subfigure}[t]{0.24\textwidth}
	\centering
        \includegraphics[width=\linewidth]{ 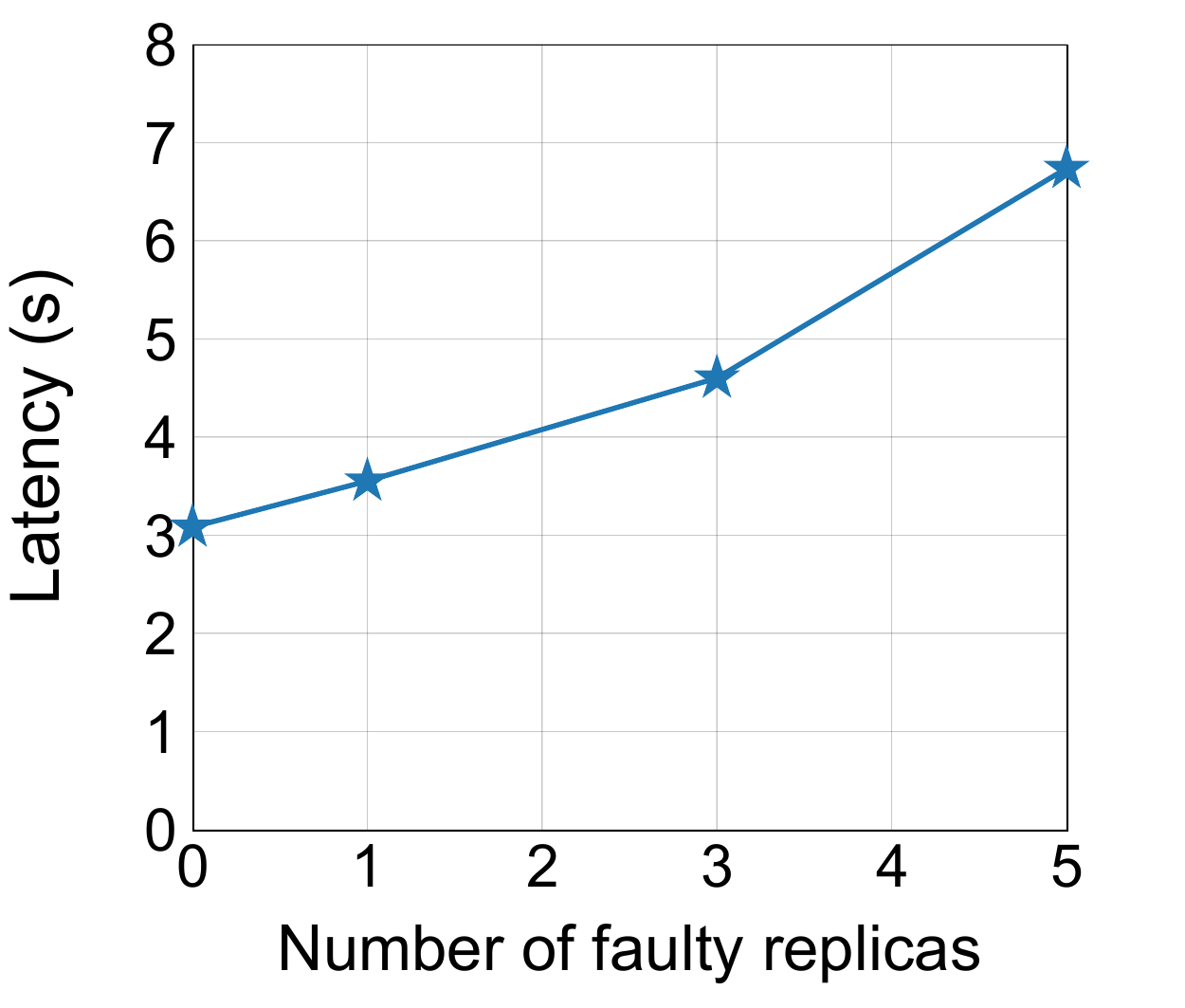}
        \caption{{Latency}}

    \end{subfigure}
    
    \caption{\textbf{{Throughput and latency of \sysname under different number of undetectable faulty replicas in WAN.}}}
	\label{fig:undetectable}
\end{figure}

\bheading{Undetectable faults.} In this section, we evaluate undetectable Byzantine behaviors, where a replica avoids participating in other instances while continuing to propose blocks in the instance it leads, thereby preventing a leader timeout. This behavior does not trigger a view change. \figref{fig:undetectable} presents the throughput and latency of \sysname under different numbers of faulty replicas in a 16-replica WAN setting. As the number of faulty replicas increases, throughput declines moderately, while latency is impacted more significantly. This is because when faults reach the maximum threshold, system latency is primarily limited by the slowest honest replica.

\section{Related Work} \label{sec:related} 
In this section, we introduce prior works of single-leader BFT consensus protocols and Multi-BFT consensus protocols. Due to space constraints, we refer readers to Appendix~\ref{appen:otherReleated} for works about the partial ordering of payment systems and other scaling approaches of BFT consensus.

\bheading{Single BFT consensus.} Castro and Liskov propose PBFT~\cite{pbft1999}, the first practical BFT protocol that 
inspires many subsequent protocols~\cite{kotla2007zyzzyva, Buchman2016TendermintBF, hotstuff, FireLedger, gueta2019sbft, liu2018scalable,lyu2023byzantine,jalalzai2023fast,jalalzai2021hermes,niu2020ebft} 
Notable examples include Zyzzyva~\cite{kotla2007zyzzyva}, Tendermint~\cite{Buchman2016TendermintBF}, and HotStuff~\cite{hotstuff}. 
These protocols adopt a leader-based scheme, which simplifies the procedure of reaching consensus on proposals, however, leads to a leader bottleneck limiting system scalability~\cite{alqahtani2021bottlenecks,charapko2021pigpaxos,gai2023scaling}.

\bheading{Multi-BFT consensus.}  
Multi-BFT consensus allows replicas to operate multiple consensus instances in parallel, effectively addressing the leader bottleneck of single BFT consensus~\cite{MIR-BFT, stathakopoulou2022state, gupta2021rcc, Ladon2025}. 
Mir-BFT~\cite{MIR-BFT} is among the first Multi-BFT protocols, which runs in epochs to globally order blocks delivered from instances by some pre-determined indices. 
An epoch change is triggered if any leader is faulty, making it vulnerable to Byzantine leaders. Later, ISS~\cite{stathakopoulou2022state} improved on Mir-BFT by permitting replicas to deliver no-op messages to prevent frequent epoch changes. 
RCC~\cite{gupta2021rcc} also adopts a pre-defined global ordering for delivered blocks. Thus, a straggler instance could delay the start of the next epoch for other replicas, leading to a slowdown of the entire system. 
DQBFT~\cite{dqbft} uses a unique BFT instance to globally order output transactions from other instances to promote efficiency. 
However, the reliance on a single ordering BFT instance can be susceptible to targeted attacks~\cite{estrada2006network}.
Ladon~\cite{Ladon2025} proposes monotonic ranks to dynamically order blocks across instances, mitigating the impact of slow replicas. These approaches reduce the latency caused by slow instances, but the global ordering phase still accounts for a significant proportion (almost half) of the system latency. 

{Unlike the above-mentioned protocols, TELL~\cite{tong2024tell} optimizes the execution layer rather than consensus layer. TELL employs a State Hash Table (SHT) to track read/write sets, enabling concurrent execution of transactions across instances and within blocks. TELL resolves conflicts through re-execution and merges instances at the epoch level, which reduces execution waiting latency but limits its overall improvement in end-to-end client latency. 
In contrast, \sysname tries to avoid cross-instance conflicts through a tailored transaction assignment strategy, eliminating the need for re-executing payment transactions to improve end-to-end latency. }

\section{Conclusion} \label{sec:conclusion}
In this paper, we introduced \sysname, a novel approach to enhance the performance of Multi-BFT protocols by sidestepping the limitations typically introduced by global ordering. 
Leveraging a hybrid design, \sysname integrates partial ordering for conflict-free transactions and global ordering for general non-commutative transactions. We also introduce a customized escrow mechanism to ensure transaction atomicity and seamless interaction between partial and global logs.  Our experiments show that in WAN of 128 replicas, \sysname decreases latency by at least 68\% as compared to pre-determined global ordering Multi-BFT protocols with one straggler. 

\normalem
\bibliographystyle{unsrt}
\bibliography{bib}

\appendices 

\section{Global Ordering Algorithm}\label{appen:globalorder}
In this section, we introduce the dynamic global ordering algorithm in Ladon~\cite{Ladon2025}, which is used in this paper. 
The dynamic global ordering algorithm requires two parameters to order blocks: the instance index and the rank. Given a set of delivered blocks with these parameters, honest replicas can independently execute the ordering algorithm to produce a consistent sequence of globally confirmed blocks without additional communication.

\begin{algorithm}[ht]
\caption{The Global Ordering Algorithm}
\label{globalalgorithm}
\begin{algorithmic}[1]  
\State \textbf{function} $\mathsf{globalOrder}(b,glog)$
\State \hspace{1.0em} $\mathsf{append}(b, \mathcal{W})$ \textcolor{purple}{//$\mathcal{W}$: blocks waiting to be confirmed}
\State \hspace{1.0em} $\mathcal{P'} \leftarrow \mathsf{getLastBlock}(\mathcal{P})$ \textcolor{purple}{//$\mathcal{P}$: delivered blocks}
\State \hspace{1.0em} $b^{*} \gets \mathsf{findLowestBlock}(\mathcal{P'})$ 
\State \hspace{1.0em} $bar = (b^{*}.rank + 1,b^{*}.index)$ \textcolor{purple}{//compute the bar}
\State \hspace{1.0em} $b{can} = \mathsf{findLowestBlock}(\mathcal{W})$ \textcolor{purple}{//find candidate block}
\State \hspace{1.0em} \textbf{while} $b_{can} \prec bar$ \textcolor{purple}{//$b_{can}$ has a lower index than $bar$}
\State \hspace{2.0em} $\mathsf{append}(b_{can}, glog)$ \textcolor{purple}{//globally confirm $b_{can}$}
\State \hspace{2.0em} $\mathcal{W} \gets \mathcal{W} \setminus b_{can}$ \textcolor{purple}{//update $\mathcal{W}$}
\State \hspace{2.0em} $b_{can} = \mathsf{findLowestBlock}(\mathcal{W})$ \textcolor{purple}{//find next $b_{can}$}
\State \hspace{1.0em} \textbf{end while}
\State \textbf{end function}

\Statex
\Statex \textcolor{purple}{//Return block with the lowest ordering index}
\State \textbf{function} ${\mathsf{findLowestBlock}}(\mathcal{V})$
\State  \hspace{1.0em} $b^{*} \gets $ first block in $\mathcal{V}$
\State \hspace{1.0em} \textbf{for} each $b\in \mathcal{V}$ \textbf{do}  
\State  \hspace{2.0em} \textbf{if} $b \prec b^{*}$
\State  \hspace{3.0em} $b^{*} \gets b$
\State \hspace{2.0em} \textbf{end if}
\State \hspace{1.0em} \textbf{end for}
\State \hspace{1.0em} \textbf{return} $b^{*}$
\State \textbf{end function}

\end{algorithmic}
\end{algorithm}

\bheading{Rank.} Before broadcasting a block, a leader first collects the highest ranks from at least $2f+1$ replicas. It then increments the highest observed rank by one and assigns this value to the proposed block. The leader then broadcasts the block along with its assigned rank. Once consensus is reached, no additional procedure is needed to ensure agreement since the rank is piggybacked with the block. 

The rank parameter has two key properties:

\begin{packeditemize}
    \item \textbf{Agreement:} All honest replicas have the same rank for a delivered block.

    \item \textbf{Monotonicity:} 
     If a block $b'$ is generated after an intra-instance (or a delivered inter-instance) block $b$, then the rank of $b'$ is larger than the rank of $b$.
\end{packeditemize}

A detailed description and formal proof can be found in \cite{Ladon2025}.

\bheading{Global ordering algorithm.}
Algorithm~\ref{globalalgorithm} shows the global ordering function running at a replica. 
The blocks are  globally ordered by increasing ranks and a tie-breaking to favor blocks with smaller instance indices. For example, given two blocks $b$ and $b^{\prime}$, block $b$ will be globally ordered before $b^{\prime}$, when $b.rank < b^{\prime}.rank$ or $b.rank = b^{\prime}.rank \wedge b.index < b^{\prime}.index$. 
For convenience, we use $b \prec b^{\prime}$ to denote that block $b$ has a lower global ordering index than block $b^{\prime}$.

When a block $b$ is delivered, the $\mathsf{globalOrder}$ function is triggered to append the block $b$ to $\mathcal{W}$, which is the set of blocks delivered by SB instances waiting to be globally ordered, \ie, blocks delivered by SB instances but have not been globally ordered. The replica then fetches the last delivered block from the set $\mathcal{P}$ of blocks delivered by SB instances, denoted by the set $\mathcal{P}^{\prime}$.
It then finds the block $b^{*}$ that has the lowest ordering index among the blocks in $\mathcal{P}^{\prime}$, \ie, $\forall b^{\prime} \in \mathcal{P}^{\prime} \mbox{ with } b^{\prime} \ne b^{*}: b^{*} \prec b^{\prime}$. 
Thereafter, \textit{bar} can be computed as a tuple of $(rank, index)$:  
\[bar := (rank, index) = (b^{*}.rank + 1, b^{*}.index).\]
The threshold $bar$ represents the lowest global ordering index that can be owned by subsequently generated blocks. The $bar$ is initialized with $(0,0)$. 

With $bar$ defined, the replica repetitively checks unconfirmed blocks and decides which blocks to confirm.  
Specifically, the replica finds the block  $b_{can} \in \mathcal{W}$ that has the lowest ordering index, which is referred to as the \textit{candidate} block.
If $b_{can}$ has a lower ordering index than $bar$, $b_{can}$ can be globally confirmed because all future blocks will have higher indices than $b_{can}$. 
In particular, $b_{can}$ will be appended to the $glog$ and removed from the set $\mathcal{W}$. 
The process repeats until no such $b_{can}$ can be found.

\color{black}

\section{Protocol Running Example}\label{appen:runexample}
To give readers a better understanding of how the protocol operates, we provide a running example of \sysname, as shown in \figref{fig:example}. Consider a system with two instances, Instance 0 and Instance 1, and three clients: Alice, Bob, and Carol, with initial balances of {\$4, \$0}, and \$0, respectively. Assume Alice is assigned to Instance 0 and Bob is assigned to Instance 1. Additionally, there is a smart contract that requires two clients to invoke it together, incurring a cost of \$1 per client.

\begin{figure}[t]
	\centering
    \includegraphics[width=1\linewidth]{ 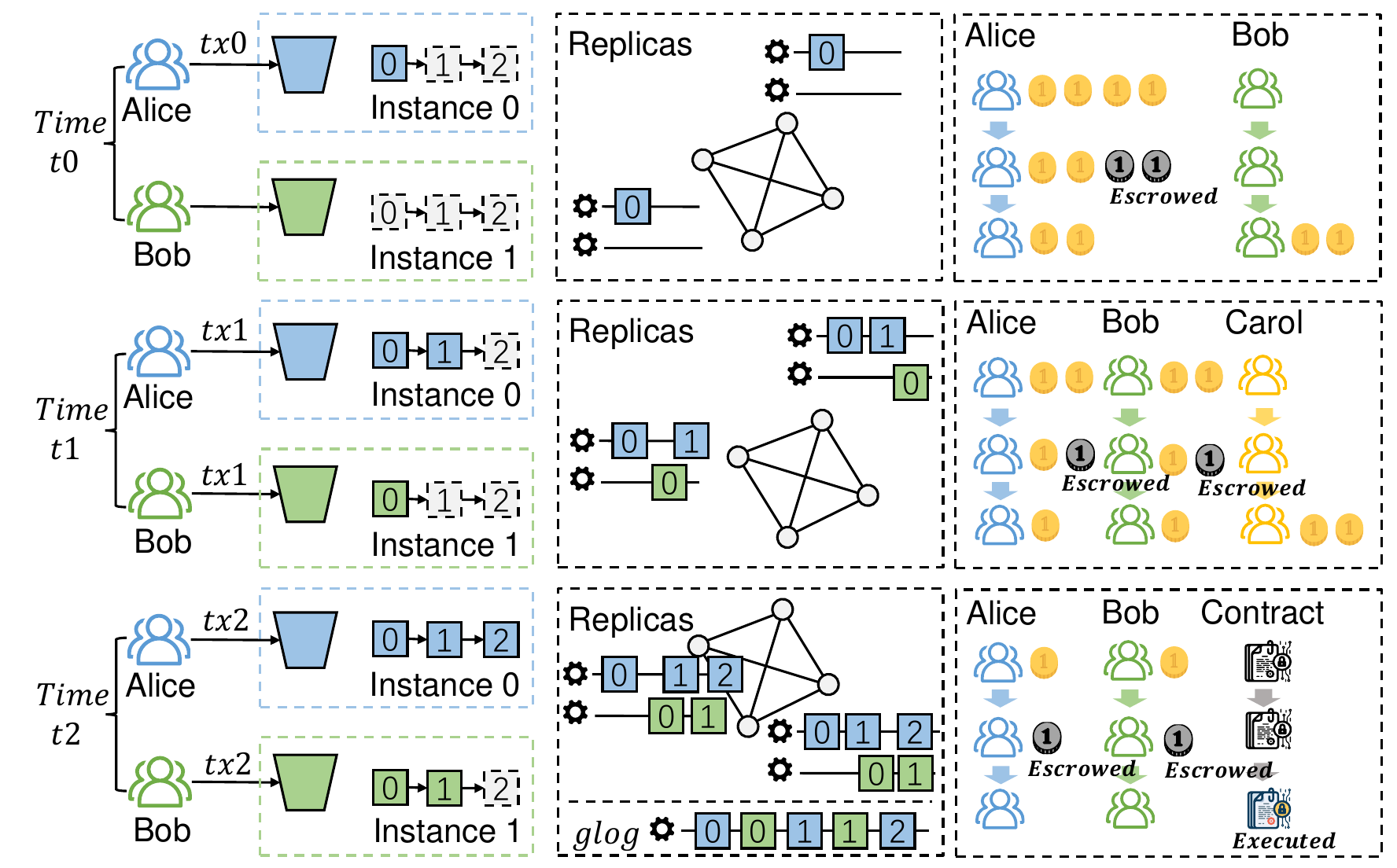}
	\caption{Protocol running example with two instances, three clients, and a contract. {$tx_i$ is a transaction arrived at time $t_i$.}}
	\label{fig:example}
\end{figure}

\bheading{Example of payment transactions.} 
{Consider a payment transaction \(tx_0\) with a single payer, Alice, and a single payee, Bob, where Alice transfers \$2 to Bob. The transaction involves two owned objects, Alice and Bob, and is assigned to Instance 0, as Alice performs a decremental operation of \$2. Assume \(tx_0\) is included in Block 0 of Instance 0. Since there is only one payer, \(tx_0\) is executed immediately, escrowing \$2 from Alice’s account and transferring it directly to Bob’s account.
}

{
Now consider another transaction \(tx_1\), where both Alice and Bob each transfer \$1 to a single payee, Carol. This transaction involves three owned objects: Alice, Bob, and Carol. Alice and Bob perform decremental operations of \$1 each, while Carol performs an incremental operation of \$2. As a result, \(tx_1\) is assigned to both Instance 0 (for Alice) and Instance 1 (for Bob).
Assume \(tx_1\) is included in Block 1 of Instance 0 and Block 0 of Instance 1. Block 0 of Instance 1 must refer to state \(S = \{0, \bot\}\), which represents Block 0 of Instance 0 (containing \(tx_0\)) and the initial state of Instance 1. This dependency ensures that Bob’s transfer relies on the completion of \(tx_0\), where Bob receives sufficient balance.
In Block 1 of Instance 0, \$1 is escrowed from Alice’s account. Concurrently, in Block 0 of Instance 1, \$1 is escrowed from Bob’s account. These operations can be processed in parallel because they belong to different instances and do not conflict. However, the escrow operation in Instance 1 must occur after \(tx_0\), as indicated by the reference to Block 0 of Instance 0.
Once both escrow operations are successfully processed, the replicas commit them, and \$2 is transferred to Carol’s account, completing the transaction \(tx_1\).
}

\bheading{Example of contract transactions.} Consider a contract transaction, \(tx_2\), where both Alice and Bob call the contract together. In this transaction, \(tx_2\) involves two owned objects: Alice and Bob. Alice and Bob each perform a decremental operation of \$1. Since \(tx_2\) involves both Alice and Bob, which are owned objects with a decremental operation, it is assigned to Instance 0 and Instance 1. 

Assume that \(tx_2\) is included in Block 2 of Instance 0 and Block 1 of Instance 1. 
In Block 2 of Instance 0, \$1 is escrowed from Alice’s account, while in Block 1 of Instance 1, \$1 is escrowed from Bob’s account. These two operations can be executed concurrently because they are payment operations (\ie, decremental operations on owned objects). 

Once a replica observes that both escrow requests have been successfully processed in their respective instances, it can execute other operations in the contract {after the previous 4 blocks in the $glog$} and commit both escrow operations. This ensures consistent execution results across different replicas, since the contract may contain shared objects that can not be confined to a certain instance.

\section{Other Related Works} \label{appen:otherReleated}
\bheading{Partial ordering of payment systems.}
CryptoConcurrency~\cite{tonkikh2023cryptoconcurrency} introduces dynamic detection of overspending based on the balance of the account, determining whether concurrent transactions can be satisfied without requiring consensus.
Astro~\cite{collins2020online}  maintains a log separately for each client to prevent double-spending. ABC~\cite{sliwinski2021asynchronous} enables validators to parallelize transaction processing without requiring consensus, increasing system efficiency. FastPay~\cite{baudet2020fastpay} takes advantage of payment semantics to minimize shared state between accounts, allowing asynchronous operations to run with greater concurrency. Flash~\cite{lewis2023flash} sidesteps the need for reliable broadcast by utilizing a partially-ordered, DAG-like structure. A non-sequential specification for money transfer objects is presented in~\cite{auvolat2020money}, which operates on a reliable broadcast abstraction. Pastro~\cite{kuznetsov2023permissionless} introduces a partially ordered transaction set that defines active participants and stake distribution, offering an adaptable approach for applications that do not require a global order. These approaches focus on payments but lack support for smart contracts.

\bheading{Other scaling approaches.} Similar to Multi-BFT consensus, DAG and sharding systems also utilize instance parallelism to enhance system scalability. 
Notable DAG protocols include DAG-Rider~\cite{keidar2021all}, Narwhal and Tusk~\cite{danezis2022narwhal}, Bullshark~\cite{spiegelman2022bullshark}, Shoal~\cite{spiegelman2024shoal}, BBCA~\cite{malkhi2023bbca} and Mysticeti~\cite{babel2025mysticeti} require each block to reference at least $2f + 1$ predecessor blocks. 
BFT sharding protocols~\cite{al2017chainspace, dang2019towards, huang2022brokerchain, hellings2021byshard, hong2024gridb} such as Elastico~\cite{Luu2016}, OmniLedger~\cite{omniledger}, and RapidChain~\cite{zamani2018rapidchain}, Sharper~\cite{amiri2021sharper}, Byshard~\cite{hellings2021byshard},  divide replicas into multiple subcommittees for parallel processing.
Except for the above approaches, some systems also employ concurrent execution of transactions to improve scalability~\cite{kotla2004high, burgos2021performance,narayanam2021partial,kang2024spotless}. 

\end{document}